\UseRawInputEncoding
\documentclass[a4paper, 10pt]{article}
\usepackage{ascmac}
\usepackage{mathrsfs}
\usepackage{latexsym}
\usepackage{amsmath}
\usepackage{amssymb}
\usepackage{amsthm}
\usepackage{comment}

\usepackage{verbatim}

\numberwithin{equation}{section}

\theoremstyle{plain}
	\newtheorem{Thm}{Theorem}[section]
	\newtheorem{Lem}[Thm]{Lemma}
	\newtheorem{Prop}[Thm]{Proposition}

\theoremstyle{definition}
	
	\newtheorem{Rem}[Thm]{Remark}
	\newtheorem{Ass}[Thm]{Assumption}
	
\newcommand{\e}{\epsilon}
\newcommand{\G}{\Gamma}

\newcommand{\Z}{\mathbb{Z}}
\newcommand{\R}{\mathbb{R}}
\newcommand{\C}{\mathbb{C}}
\newcommand{\N}{\mathbb{N}}

\newcommand{\F}{\mathcal{F}}

\newcommand{\im}{\mathrm{i}}
\newcommand{\ran}{\mathrm{ran}}
\newcommand{\ind}{\mathrm{ind}}
\newcommand{\sgn}{\mathrm{sgn}}

\DeclareMathOperator{\Tr}{Tr}

\usepackage{float}
\usepackage{tikz}
	\usetikzlibrary{arrows}
\usepackage{pgfplots}
	\pgfplotsset{compat=1.12}
\usepackage{hyperref}

\makeatletter
\renewcommand*{\eqref}[1]{%
  \hyperref[{#1}]{\textup{\tagform@{\ref*{#1}}}}%
}
\makeatother

\usepackage[right]{lineno}

\begin{document}
\title{The Witten index for one-dimensional split-step quantum walks under the non-Fredholm condition}

\author {Yasumichi Matsuzawa, Akito Suzuki, Yohei Tanaka, Noriaki Teranishi, Kazuyuki Wada}

\maketitle

AMS 2020 Classification: 47A53, \ 81Q35, \ 81Q60
\\
Keywords: Quantum walk, Witten index, Chiral symmetry, Spectral shift function

\begin{abstract}

It is recently shown that a split-step quantum walk possesses a chiral symmetry, and that a certain well-defined index can be naturally assigned to it. 
The index is a well-defined Fredholm index if and only if the associated unitary time-evolution operator has spectral gaps at both $+1$ and $-1.$
In this paper we extend the existing index formula for the Fredholm case to encompass the non-Fredholm case (i.e., gapless case). We make use of a natural extension of the Fredholm index to the non-Fredholm case, known as the Witten index. 
The aim of this paper is to fully classify the Witten index of the split-step quantum walk by employing the spectral shift function for a rank one perturbation of a fourth order difference operator. It is also shown in this paper that the Witten index can take half-integer values in the non-Fredholm case.
\end{abstract}

\section{Introduction and main result}
Quantum walk is widely known as a quantum mechanical counterpart of random walk \cite{MR2120299, MR958911}.  
There are various applications such as quantum computations \cite{MR2507892,PhysRevA.81.042330}, quantum algorithms \cite{MR2306290, MR2318713,PhysRevA.67.052307} and topological phases \cite{PhysRevB.88.121406,PhysRevB.82.235114,PhysRevA.82.033429}.
It is also regarded as a discretization of the Dirac equation \cite{MR4092112}.

One of the most important properties of quantum walks is localization.
In particular, the famous Grover search algorithm \cite{MR1427516} deeply depends on the property.
Localization, however, does not always occur.
A sufficient condition for the occurrence of localization is that the time evolution operator has an eigenstate and has no singular continuous spectrum \cite[Proposition 2.4]{MR3488238}.
The absence of singular continuous spectrum is shown by using scattering theory \cite[Theorem 2.4]{MR3748367}.
On the other hand, the existence of an eigenstate can be understood from the point of view of topological phases in some situations.

In this paper, we focus on mathematical studies of topological phases for discrete-time chiral symmetric quantum walks on the integer lattice $\mathbb{Z},$ which are directly related to index theory.
The point is that eigenstates (known as edge states) exist when the index is non-zero. 
This phenomena is called the bulk-edge correspondence.
Such systems have been experimentally realized \cite{PhysRevA.96.033846,kitagawa2012observation} and mathematically studied \cite{MR3748295,MR4205231,endo2015relation,MR3842302,MR4031312}.
The above phenomena corresponds to the so-called gapped phases in physics.
To explain more precisely, we briefly recall some basic facts of chiral symmetric quantum walks from \cite{MR4031312} (see also \cite{MR3748295}).
Quantum walks we consider in this paper are discrete-time quantum walks.
That is, the (unit) time evolution of a system is described by a unitary operator $U$ on a Hilbert space $\mathcal{H}$.
For a unitary and self-adjoint operator $\G$ on $\mathcal{H}$, we say that $U$ has a chiral symmetry $\Gamma,$ if the following equality holds;
\begin{align}
\label{chiral}
\G U\G=U^{\ast}.
\end{align}
If $U$ has a chiral symmetry $\Gamma,$ then the imaginary part $Q$ of $U$ plays a role of a supercharge in the context of supersymmetric quantum mechanics \cite[Section 5.1.2]{MR1219537};
$$Q:=\displaystyle\frac{U-U^{\ast}}{2\im}.$$
That is, $Q$ anti-commutes with $\G$, meaning that $\G Q+Q\G=0$ holds. 
From this, $Q$ is decomposed as 
\begin{equation}\label{intro Q_0}
Q = 
\begin{bmatrix}
0 & Q_{0}^{\ast} \\ 
Q_{0}& 0
\end{bmatrix}_{\ker(\G-1) \oplus \ker(\G+1)},
\qquad Q_0:\ker(\G-1) \to \ker(\G+1).
\end{equation}
If $Q_{0}$ is Fredholm, we say that the pair $(\Gamma, U)$ is Fredholm.
In this case, we can introduce an index $\ind_{\Gamma}(U)$ of the pair $(\Gamma, U)$ by
\begin{equation}
\label{fredholmindex}
\begin{aligned}
\ind_\G(U):&=\dim \ker(Q_{0})-\dim \ker(Q_{0}^{\ast}),
\end{aligned}
\end{equation}
which is nothing but the Fredholm index of $Q_0.$
For Fredholm index theory, see e.g. \cite[Section 1.4]{MR1074574} or \cite{MR1369029}.
We note that the index $\ind_{\G}(U)$ depends not only on $U,$ but also on $\G.$ Thus we emphasize dependence on $\G$ in the suffix of $\ind(\cdot).$ 
The Fredholm condition of the pair $(\Gamma, U)$ is equivalent to that neither $+1$ nor $-1$ is contained in the essential spectrum of $U.$
Thus, in this case, the spectrum of $U$ has gaps at both $+1$ and $-1$, and hence we can interpret that this case corresponds to a gapped phase in physics.
In addition, it is established in \cite[Theorem 3.4]{MR4031312} (see also \cite[Corollary 4.3]{MR3748295}) that the following inequality holds;
$$\dim \ker(U-1)+\dim \ker(U+1)\ge |\ind_{\G}(U)|.$$
The above inequality implies that if $\ind_{\G}(U)\neq0 ,$ then there exists an 
eigenstate corresponding to $+1$ or $-1.$ 
This is a mathematical formulation of the bulk-edge correspondence.

Motivated by these facts, the indices for split-step quantum walks are discussed in \cite{MR4119516, MR4031326, MR4219958}. 
Split-step quantum walks are originally introduced by Kitagawa et.~al \cite{PhysRevA.82.033429} to investigate topological phases of quantum walks. 
Later, Suzuki et.~al. \cite{MR3842302} generalize Kitagawa's split-step quantum walks, which we shall consider in this paper.
To explain more precisely, let us recall the definition of split-step quantum walks.
Let $\ell^2(\Z)$ be the Hilbert space of square-summable $\C$-valued sequences indexed by the set $\Z$ of integers, and let $\mathcal{H}:=\ell^2(\Z)\oplus \ell^2(\Z).$
Let $L$ be the left-shift operator on $\ell^2(\Z),$ that is,
\[
(L\psi)(x):=\psi(x+1),\quad \psi\in \ell^2(\Z), \quad x\in\Z.
\]
We give a data $a,b,p$ and $q$, where $a,p$ are two $\R$-valued sequences and $b,q$ are two $\C$-valued sequences, satisfying $a(x)^2 + |b(x)|^2 = 1$ and $p(x)^2 + |q(x)|^{2} = 1$ for any $x \in \Z.$ 
We identify the sequences with the corresponding multiplication operators on $\ell^2(\mathbb{Z}).$
The time evolution operator $U$ of a split-step quantum walk is a unitary operator on $\mathcal{H}$ defined by
\[
U:=\Gamma \Gamma'
\]
with
\begin{align}
\Gamma:=\begin{bmatrix}
p & qL \\ L^*q^{\ast} & -p(\cdot-1)
\end{bmatrix},
\qquad 
\Gamma':=
\begin{bmatrix}
a & b^{\ast} \\ b & -a
\end{bmatrix}.
\end{align}
We note that $U$ has a chiral symmetry $\Gamma$ (see e.g. \cite[Lemma 2.1]{MR4031312}), and thus the operator $Q_0$ is defined via the decomposition \eqref{intro Q_0}.
In addition, if $Q_0$ is Fredholm, the index $\textrm{ind}_\Gamma(U)$ is defined as well.
To see an explicit form of $Q_0$, we make use of the following proposition.

\begin{Prop}[{\cite[Lemma 3.2]{MR4219958}}]
There exists a unitary operator $\epsilon$ on $\mathcal{H}$
such that
\begin{align*}
&\epsilon^{\ast} \G\epsilon =
 \begin{bmatrix}
 1 & 0 \\ 
 0 & -1
 \end{bmatrix}_{\ell^2(\Z)\oplus \ell^2(\Z)}, & 
&\epsilon^{\ast} Q\epsilon = 
 \begin{bmatrix}
 0 & Q_{\e_0}^* \\ 
 Q_{\e_0} & 0
 \end{bmatrix}_{\ell^2(\Z)\oplus \ell^2(\Z)},
\end{align*}
where $Q_{\epsilon_{0}}$ is a bounded operator on $\ell^2(\mathbb{Z})$ defined by
\begin{equation}\label{ymq}
Q_{\e_0}:=\displaystyle\frac{\im}{2}\left[\sqrt{1+p} e^{\im\theta}L b^*\sqrt{1+p}-\sqrt{1-p}b L^*e^{-\im\theta} \sqrt{1-p}-|q|\{a+a(\cdot+1)\}\right]
\end{equation}
and where $\theta$ is a $\mathbb{R}$-valued sequence satisfying $q(x)=|q(x)|e^{\im\theta(x)}$ for any $x\in\mathbb{Z}.$
Moreover, $(\Gamma, U)$ is Fredholm if and only if $Q_{\e_0}$ is Fredholm.
In this case, the index $\ind_\Gamma(U)$ equals to the Fredholm index of $Q_{\e_0}.$
\end{Prop}

Since $Q_{\e_0}$ is concrete enough, the above proposition enable us to compute the index explicitly under suitable assumptions.

\begin{Thm}[{\cite[Theorem B]{MR4219958}}]\label{old result}
Suppose that the following eight limits exist:
\begin{equation}
\label{equation: anisotropic assumption}
\spadesuit_\pm := \lim_{x\to\pm\infty} \spadesuit(x), \qquad \spadesuit\in\{a,b,p,q\}.
\end{equation}
Then $(\Gamma, U)$ is Fredholm if and only if $|a_+|\not=|p_+|$ and $|a_-|\not=|p_-|$.
In this case, we have
\[
\ind_\Gamma(U)=
\begin{cases}
0, & |p_-|<|a_-|\ \text{and}\  |p_+|<|a_+|,\\
+\sgn\,p_+, & |p_-|<|a_-|\ \text{and}\  |p_+|>|a_+|,\\
-\sgn\,p_-, & |p_-|>|a_-|\ \text{and}\  |p_+|<|a_+|,\\
+\sgn\,p_+-\sgn\,p_-, & |p_-|>|a_-|\ \text{and}\  |p_+|>|a_+|,
\end{cases}
\]
where the sign function $\sgn:\mathbb{R}\to\{-1,0,+1\}$ is defined by
\begin{equation}\label{sgn_def}
\sgn\, x :=
\begin{cases}
+1, & x> 0,\\
0,& x=0,\\
-1, & x<0.
\end{cases}
\end{equation}

\end{Thm}
Cedzich et.~al.~\cite{MR3748295} discuss the symmetry index of essentially gapped unitary operators with suitable symmetries from a representation theoretical point of view. Moreover, they discuss split-step quantum walks as a concrete example.

Recently, gapless phases have been studied by physicists \cite{li2022symmetry,PhysRevX.7.041048,PhysRevB.104.075132,PhysRevLett.120.057001,PhysRevX.11.041059}.
Its mathematical understanding, however, has not been fully elucidated.
The aim of this paper is to extend the above index formula to include the gapless case.
Our motivation consists of two parts.
The first is to develop the theory of gapless phases in a mathematically rigorous way.
Especially, we would like to introduce a suitable \textit{index} for such systems and clarify the physical meaning of the index.
The second is to study gapless phases for quantum walks.
This is because quantum walks can be experimentally realized, and thus we expect that gapless phases can be experimentally realized by making use of quantum walks.
As a first step of our project, we shall achieve the aim we have already mentioned.

We note that the gapless case corresponds to the non-Fredholm case since the gapped case corresponds to the Fredholm case.
One of the most natural extensions of the Fredholm index to the non-Fredholm case is the so-called Witten index.
Let us briefly recall the theory of the Witten index.
Let $A$ be a bounded operator on a Hilbert space.
We suppose the trace class condition for $A$, meaning that $A^*A-AA^*$ is of trace class.
Then the operator $e^{-tA^*A}-e^{-tAA^*}$ is of trace class as well for all $t>0$, and we set
\[
\ind_t(A):=\Tr(e^{-tA^*A}-e^{-tAA^*}),\qquad t>0.
\]
The Witten index $w(A)$ of $A$ is defined by
\[
w(A):= \lim_{t\to\infty}\ind_t(A)
\]
whenever the limit exists.
It is known that the Witten index is equal to the Fredholm index if $A$ is Fredholm \cite[Theorem 2.5]{MR950085}.

We now return to split-step quantum walks.
Since the index $\ind_\Gamma (U)$ is equal to the Fredholm index of $Q_{\epsilon_0}$ if  the pair $(\Gamma, U)$ is Fredholm, we would like to compute the Witten index of $Q_{\epsilon_0}$, which is the main result of this paper. To state the result, we introduce the following assumption in addition to (\ref{equation: anisotropic assumption}):
\begin{equation}
\label{equation: trace assumption}
\sum_{x\ge 0}
\Bigl[|\spadesuit(x)-\spadesuit_{+}|+|\spadesuit(-x-1)-\spadesuit_{-}|\Bigr] < \infty, \qquad \spadesuit\in\{a,b,p,q\}.
\end{equation}


The above summable conditions are needed to guarantee the trace class condition for $Q_{\epsilon_0}.$ 
We are now in a position to state the main result.

\begin{Thm}\label{thm:main}
We assume the existence of limits of the form (\ref{equation: anisotropic assumption}) together with the summable condition (\ref{equation: trace assumption}).
Then we have
\[
w(Q_{\epsilon_0}) = W(a_+,p_+)-W(a_-,p_-),
\]
where the function $W:[-1,1]\times[-1,1]\to\mathbb{R}$ is defined by
\[
W(r,s):=\begin{cases}
0,\qquad &|r|=|s|=1,\\
[\mathrm{sgn}(r+s)-\mathrm{sgn}(r-s)]/2, &\text{otherwise}.
\end{cases}
\]
\end{Thm}

We note that the above formula coincides with Theorem \ref{old result} if $Q_{\epsilon_0}$ is Fredholm.
This follows from a direct computation.

As a consequence of the main result, we classify the Witten index and show that these indices only depend on the limits on both sides. Moreover, it takes half-integers such as $3/2, 1/2, -1/2$ or $-3/2.$ The appearance of half-integer values suggests the possibility of the existence of resonances at $+1$ or $-1.$ 
In other words, there may exist $\psi$ which the solution of eigenvalue problems
$$ U\psi= \psi,\quad \text{ or }\quad U\psi=-\psi$$
and $\psi\notin \ell^2(\mathbb{Z})\oplus\ell^2(\mathbb{Z})$ but $\psi\in \ell^\infty(\mathbb{Z})\oplus\ell^\infty(\mathbb{Z}).$ 
But the problem is left for the future study.
For quantum mechanical cases, see e.g. \cite{MR894842}.

Our proof is based on computing spectral shift functions and perturbation determinants directly,
	which is different from the method used in \cite{MR894842}.
For spectral shift functions and perturbation determinants, we refer \cite[Chapter 9]{MR2953553}.
The main difference between \cite{MR894842} and the present paper is as follows. 
In \cite{MR894842}, the authors mainly consider Dirac operators of the form 
\begin{align}\label{dirac}
Q:=\begin{bmatrix}0 & A^{\ast} \\ A & 0\end{bmatrix}_{L^2(\mathbb{R})\oplus L^2(\mathbb{R})},\quad A:=\displaystyle\frac{d}{dx}+\phi,\quad \phi:\mathbb{R}\rightarrow \R,
\end{align}
and compute the Witten index of $A.$
Since both $A^{\ast}A$ and $AA^{\ast}$ are second order differential operators whose properties are well known, the Witten index of \eqref{dirac} can be calculated by constructing Jost solutions. 
On the other hand, in our case, $Q_{\e_0}$ itself is  a second order difference operator.
Thus, to derive the Witten index of $Q_{\e_0}$, we have to deal with fourth order difference operators $Q_{\e_0}^*Q_{\e_0}$ and $Q_{\e_0}Q_{\e_0}^*$ that is more complicated than second order one.
Indeed, we face the resolvent of a fourth order difference operator when we derive the perturbation determinant (see \eqref{def:pd} below). 

Let us explain our strategy to compute the Witten index of $Q_{\e_0}.$
First, we cut the lattice $\mathbb{Z}$ at the origin, and reduce the computation of $Q_{\e_0}$ to that of some Toeplitz operator.
Second, 
We can reduce the computation of the spectral shift function of the pair $(Q_{\e_0}^*Q_{\e_0},Q_{\e_0}Q_{\e_0}^*)$
to that of the pair $(T,T_0)$, where $T_0$ is a Toeplitz operator and $T$ is a rank one perturbation of $T_0.$

The rest of this paper is organized as follows. 
In Section \ref{section: section 3}, we reduce $Q_{\epsilon_0}$ to a more simplified form. In Section \ref{section: section 4}, we analyze the properties of the perturbation determinant coming from $Q_{\epsilon_0}.$ In Section \ref{section: section 5}, we derive the spectral shift function and prove the main result. In Appendix \ref{section: Appendix A}, we prove the trace property of the difference between two exponential operators. In  Appendix \ref{section: Appendix B}, we discuss the spectral decomposition of the discrete Laplacian on the half-line. In  Appendix \ref{section: Appendix C}, we prove several properties of a quadratic equation which is closely related to Section \ref{section: section 4} and \ref{section: section 5}. 
In Appendix \ref{section: Appendix D}, we eliminate the phase terms of $b$ and $q$ by a unitary transformation to avoid the difficulty arising from the phase terms.
Contents of  Appendices \ref{section: Appendix A}, \ref{section: Appendix B} and \ref{section: Appendix C} may be well-known. Nevertheless, we include these supplementary sections for completeness.

Throughout the paper, we always assume the assumption of Theorem \ref{thm:main}.

\section{Reduction of \texorpdfstring{$Q_{\epsilon_{0}}$}{Qepsilonzero}}
\label{section: section 3}
In this section, we simplify $Q_{\epsilon_{0}}$ by using the summable conditions \eqref{equation: trace assumption}.
Let $\mathcal{B}(\mathcal{X})$ be the set of bounded operators on a Hilbert space $\mathcal{X}.$ 
We denote the set of trace class (resp. Hilbert-Schmidt) operators on $\mathcal{X}$ by $\mathcal{S}_{1}(\mathcal{X})$ (resp. $\mathcal{S}_{2}(\mathcal{X})$).
Let us recall the invariance of $\ind_{t}(\cdot)$ under trace class perturbations:
\begin{Thm}
\label{GS}
We define $\tilde{A}:=A+C$ with $A\in\mathcal{B}(\mathcal{X})$ and $C\in\mathcal{S}_{1}(\mathcal{X}).$ Then
\begin{align}
\label{equiv} A^* A-AA^*\in \mathcal{S}_{1}(\mathcal{X}) \mbox{ if and only if } \tilde{A}^* \tilde{A}-\tilde{A}\tilde{A}^*\in\mathcal{S}_{1}(\mathcal{X}).
\end{align}
In this case, we have 
\begin{align}\label{ind:equal}
\ind_{t}(\tilde{A})=\ind_{t}(A).
\end{align}
\end{Thm}
\begin{proof}
Since $\mathcal{S}_{1}(\mathcal{X})$ is an ideal, \eqref{equiv} follows from direct calculations. 
\eqref{ind:equal} follows from an application of \cite[Theorems 3.1 and 3.2]{MR950085}. 
\end{proof}
The following proposition easily follows from the properties of the trace.
\begin{Prop}\label{w}
The following assertions hold:
\begin{enumerate}
\item Suppose that for $A\in\mathcal{B}(\mathcal{X}),$ $w(A)$ exists. Then for any unitary operator $G,$ $w(GAG^{-1})$ exists and $w(GAG^{-1})=w(A).$
\item Let $A_{1}\in\mathcal{B}(\mathcal{X}_{1})$ and $A_{2}\in\mathcal{B}(\mathcal{X}_{2}),$
where $\mathcal{X}_1$ and $\mathcal{X}_2$ are two Hilbert spaces.
Suppose that both $w(A_{1})$ and $w(A_{2})$ exist. Then $w(A_{1}\oplus A_{2})$ exists and $w(A_{1}\oplus A_{2})=w(A_{1})+w(A_{2}).$
\end{enumerate}
\end{Prop}

We now return to split-step quantum walks.
First of all, by Theorem \ref{lemma: NOW transform} and Proposition \ref{ind:equal}, we may assume without loss of generality that $b(x)\geq 0$ and $q(x)\geq0$ for all $x\in\mathbb{Z}.$
Thus we can choose $\theta$ in \eqref{ymq} so that $\theta(x)=0$ for any $x\in\mathbb{Z}.$
Next, we decompose $\ell^2(\Z)$ into two closed subspace $\ell^2(\Z_{\ge 0})$ and $\ell^2(\Z_{\le -1}),$ where
\begin{align}\label{Hpm}
\ell^2(\Z_{\ge 0}):=\{\psi\in \ell^2(\Z)|\ \psi(x)=0,\ x\le -1\},\quad \ell^2(\Z_{\le -1}):=\{\psi\in \ell^2(\Z)|\ \psi(x)=0,\ x\ge 0\}.
\end{align}
We denote orthogonal projections from $\ell^2(\Z)$ onto $\ell^2(\Z_{\ge 0})$ and $\ell^2(\Z_{\le -1})$ by $P_{+}$ and $P_{-},$ respectively. Moreover, we define an operator $Q_{\e_{0}, 0}$ by
\begin{align}\label{Qd}
Q_{\epsilon_{0}, 0}:=P_{-}Q_{\epsilon_{0}}P_{-}+ P_{+}Q_{\epsilon_{0}}P_{+}.
\end{align}
\begin{Lem}\label{cut:origin}
We have $Q_{\epsilon_{0}}-Q_{\e_{0},0}\in\mathcal{S}_{1}(\ell^2(\Z)).$
\end{Lem}
\begin{proof}
For any $\psi\in \ell^2(\Z),$ we see that
$$((Q_{\e_{0}}-Q_{\e_{0}, 0})\psi)(x)=0,\quad x\in\Z \setminus\{-1, 0\}.$$
Thus, the rank of $Q_{\epsilon_{0}}-Q_{\e_{0},0}$ is at most two.
\end{proof}
We introduce the following functions $F_+$ and $F_-$:
\begin{align}\label{f:pm}
F_{\pm}(z):=\displaystyle\frac{\im}{2}\left[ (1+p_{\pm})b_{\pm}z-(1-p_{\pm})b_{\pm}\overline{z}-2q_{\pm}a_{\pm}\right],
\qquad z\in\mathbb{C}.
\end{align}
For a bounded operator $X$ on a Hilbert space $\mathcal{X}$, we set
\[
F_\pm(X) := \frac{\im}{2}\left[ (1+p_{\pm})b_{\pm}X-(1-p_{\pm})b_{\pm}X^*-2q_{\pm}a_{\pm}\right].
\]

\begin{Lem}\label{two-phase}
We have
$$P_{+}Q_{\epsilon_{0}}P_{+}-P_{+}F_{+}(L)P_{+}\in\mathcal{S}_{1}(\ell^2(\Z_{\ge 0})),\quad P_{-}Q_{\epsilon_{0}}P_{-}-P_{-}F_{-}(L)P_{-}\in\mathcal{S}_{1}(\ell^2(\Z_{\le -1})).$$
\end{Lem}
\begin{proof}
We only consider the difference $P_{+}Q_{\epsilon_{0}}P_{+}-P_{+}F_{+}(L)P_{+}.$ The other case is similar to this. First, we decompose $Q_{\e_{0}}-F_{+}(L)$ into three parts:
\begin{align*}
-2\im P_+\left(Q_{\e_{0}}-F_{+}(L)\right)P_+&=P_+\left\{\sqrt{1+p}Lb\sqrt{1+p}-(1+p_{+})b_{+}L\right\}P_+\\
&\hspace{5mm}-P_+\left\{\sqrt{1-p}bL^{\ast}\sqrt{1-p}-(1-p_{+})b_{+}L^{\ast}\right\}P_+\\
&\hspace{10mm}-P_+\left\{q(a+a(\cdot+1))-2q_{+}a_{+}\right\}P_+.
\end{align*}
The term $q(a+a(\cdot+1))-2q_{+}a_{+}$ can be deformed as
$$ q(a+a(\cdot+1))-2q_{+}a_{+}=(q-q_{+})(a+a(\cdot+1))+q_{+}(a-a_{+})+q_{+}(a(\cdot+1)-a_{+}).$$
By the summable conditions \eqref{equation: trace assumption}, 
we have
$P_{+}\{q(a+a(\cdot+1))-2q_{+}a_{+}\}P_{+}\in\mathcal{S}_{1}(\ell^2(\Z_{\ge 0})).$

Next, we consider the first term. 
For the multiplication operator by a bounded sequence $m:\mathbb{Z}\rightarrow\C,$ it follows that $Lm=m(\cdot+1)L$ and $L^{\ast}m=m(\cdot-1)L^{\ast}.$ 
Thus, we have
\begin{align*}
&\hspace{5mm}\sqrt{1+p}Lb\sqrt{1+p}-(1+p_{+})b_{+}L
\\
&=\sqrt{1+p}\sqrt{1+p(\cdot+1)}Lb-(1+p_{+})b_{+}L
\\
&=\left\{\sqrt{1+p}\sqrt{1+p(\cdot+1)}-(1+p_{+})\right\}Lb+(1+p_{+})L(b-b_{+}).
\end{align*}
It is seen that $P_{+}\{(1+p_{+})L(b-b_{+})\}P_{+}\in\mathcal{S}_{1}(\ell^2(\Z_{\ge 0})).$ 
We show $P_{+}\{\sqrt{1+p}\sqrt{1+p(\cdot+1)}-(1+p_{+})\}P_{+}\in\mathcal{S}_{1}(\ell^2(\Z_{\ge 0})).$ If $p_{+}\in(-1, 1],$ it can be deformed as follows:
\begin{align*}
\left|\sqrt{1+p}\sqrt{1+p(\cdot+1)}-(1+p_{+})\right|&=\left|\displaystyle\frac{p+p(\cdot+1)+p\cdot p(\cdot+1)-2p_{+}-p_{+}^2}{\sqrt{1+p}\sqrt{1+p(\cdot+1)}+(1+p_{+})}\right|
\\ 
&=\left|\displaystyle\frac{(p-p_{+})+\{p(\cdot+1)-p_{+}\}+\{p\cdot p(\cdot+1)-p_{+}^2\}}{\sqrt{1+p}\sqrt{1+p(\cdot+1)}+(1+p_{+})}\right|
\\ 
&\le \left|\displaystyle\frac{(p-p_{+})+\{p(\cdot+1)-p_{+}\}+\{p\cdot p(\cdot+1)-p_{+}^2\}}{1+p_{+}}\right|.
\end{align*}
By the summable conditions \eqref{equation: trace assumption}, we have $P_{+}\{\sqrt{1+p}\sqrt{1+(\cdot+1)}-(1+p_{+})\}P_{+}\in S_{1}(\ell^2(\Z_{\ge 0})).$ 
If $p_{+}=-1,$ the summable conditions \eqref{equation: trace assumption} imply
\begin{align}
\displaystyle\sum_{x\ge0}|p(x)-p_{+}|=\displaystyle\sum_{x\ge0}(1+p(x))<\infty.
\end{align}
Moreover, it is seen that
\begin{align}
\sqrt{1+p(x)}\sqrt{1+p(x+1)}-(1+p_+)=\sqrt{1+p(x)}\sqrt{1+p(x+1)},\quad x\ge 0.
\end{align}
Since $\sqrt{1+p},$ $\sqrt{1+p(\cdot+1)}\in \mathcal{S}_{2}(\ell^2(\Z_{\ge 0})),$ we have $\sqrt{1+p}\sqrt{1+p(\cdot+1)}\in \mathcal{S}_{1}(\ell^2(\Z_{\ge 0})).$ 
Therefore the first term is in $\mathcal{S}_{1}(\ell^2(\Z_{\ge 0})).$

It can be proven similarly that the second term is in $\mathcal{S}_{1}(\ell^2(\Z_{\ge 0})).$
This completes the proof.
\end{proof}
By Lemma \ref{cut:origin} and Lemma \ref{two-phase}, we have the following lemma.
\begin{Lem}
\label{reduce}
It follows that $Q_{\epsilon_{0}}-\{P_{+}F_{+}(L)P_{+}+ P_{-}F_{-}(L)P_{-}\}\in\mathcal{S}_{1}(\ell^2(\Z)).$
\end{Lem}
Before going to the next proposition, we introduce several notations. We set $\{\delta_{x}\}_{x\ge 0}$ as the standard basis of $\ell^2(\Z_{\ge0})$ and $v$ is the right-shift operator on $\ell^2(\Z_{\ge0})$, that is
\[
v\delta_{x}=\delta_{x+1},\qquad x\geq0.
\]
Let $\Omega_{0}$ be the orthogonal projection from $\ell^2(\Z_{\ge0})$ onto $\mathbb{C}\delta_0$.
Then it follows that
\begin{align}\label{halfline}
\quad v^{\ast}v=1, \quad vv^{\ast}=1-\Omega_{0}.
\end{align}
Let $\{|x\rangle\}_{x\in\mathbb{Z}}$ be the standard basis of $\ell^2(\mathbb{Z}).$
We define a unitary operator $V:\ell^2(\Z)\rightarrow \ell^2(\Z_{\ge0})\oplus \ell^2(\Z_{\ge0})$ by
\begin{align}\label{v}
	V|x\rangle := 
	\begin{cases}
		(0, \delta_{x}),  &x\ge 0, \\ 
		(\delta_{-x-1}, 0), &x\le -1.
	\end{cases}
\end{align}

\begin{Prop}\label{cut}
	It follows that
$$V\{P_{-}F_{-}(L)P_{-}+P_{+}F_{+}(L)P_{+}\}V^{-1}=F_{-}(v)\oplus F_{+}(v^{\ast}).$$
\end{Prop}
\begin{proof}
By the definitions of $L$ and $v,$ we have
\begin{align}
VP_{+}LP_{+}V^{-1}=0\oplus v^{\ast},\quad VP_{+}L^{\ast}P_{+}V^{-1}=0\oplus v,\quad VP_{-}LP_{-}V^{-1}=v\oplus 0,\quad VP_{-}L^{\ast}P_{-}V^{-1}=v^{\ast}\oplus 0.
\end{align}
Thus, it follows that
\begin{align*}
&V\{P_{-}F_{-}(L)P_{-}+P_{+}F_{+}(L)P_{+}\}V^{-1}
\\
&=\displaystyle\frac{\im}{2}\left[ (1+p_{-})b_{-}VP_{-}LP_{-}V^{-1}-(1-p_{-})b_{-}VP_{-}L^{\ast}P_{-}V^{-1}-2q_{-}a_{-}VP_{-}V^{-1}\right]
\\
&\hspace{10mm}+\displaystyle\frac{\im}{2}\left[ (1+p_{+})b_{+}VP_{+}LP_{+}V^{-1}-(1-p_{+})b_{+}VP_{+}L^{\ast}P_{+}V^{-1}-2q_{+}a_{+}VP_{+}V^{-1}\right]
\\
&=F_{-}(v)\oplus F_{+}(v^{\ast}),
\end{align*}
where we used that $VP_{-}V^{-1}=1\oplus 0$ and $VP_{+}V^{-1}=0\oplus 1.$
\end{proof}

By Theorem \ref{GS}, Proposition \ref{w}, \ref{cut}, and Lemma \ref{reduce}, we have
\begin{align}\label{index}
w(Q_{\epsilon_0}) = w(F_{+}(v^{\ast}))+w(F_{-}(v)),
\end{align}
whenever $w(F_{+}(v^{\ast}))$ and $w(F_{-}(v))$ exist.
\begin{Lem}\label{lem:fastf}
We have
\begin{align*}
F_{+}(v^{\ast})^{\ast}F_{+}(v^{\ast})&=\displaystyle\frac{1}{4}\left[-b_{+}^2(1-p_{+}^2)\{v^{2}+(v^{\ast})^2\}-4a_{+}b_{+}p_{+}q_{+}(v+v^{\ast})+2b_{+}^2(1+p_{+}^2)+4a_{+}^2q_{+}^2-b_{+}^2(1+p_{+})^2\Omega_{0}\right],
\\
F_{+}(v^{\ast})F_{+}(v^{\ast})^{\ast}&=\displaystyle\frac{1}{4}\left[-b_{+}^2(1-p_{+}^2)\{v^{2}+(v^{\ast})^2\}-4a_{+}b_{+}p_{+}q_{+}(v+v^{\ast})+2b_{+}^2(1+p_{+}^2)+4a_{+}^2q_{+}^2-b_{+}^2(1-p_{+})^2\Omega_{0}\right],
\\
F_{-}(v)^{\ast}F_{-}(v)&=\displaystyle\frac{1}{4}\left[-b_{-}^2(1-p_{-}^2)\{v^{2}+(v^{\ast})^2\}-4a_{-}b_{-}p_{-}q_{-}(v+v^{\ast})+2b_{-}^2(1+p_{-}^2)+4a_{-}^2q_{-}^2-b_{-}^2(1-p_{-})^2\Omega_{0}\right],
\\
F_{-}(v)F_{-}(v)^{\ast}&=\displaystyle\frac{1}{4}\left[-b_{-}^2(1-p_{-}^2)\{v^{2}+(v^{\ast})^2\}-4a_{-}b_{-}p_{-}q_{-}(v+v^{\ast})+2b_{-}^2(1+p_{-}^2)+4a_{-}^2q_{-}^2-b_{-}^2(1+p_{-})^2\Omega_{0}\right].
\end{align*}
In particular, $F_{+}(v^{\ast})$ and $F_{-}(v)$ satisfy the trace class condition. 
That is, we have 
\begin{align}\label{compartible}
F_{+}(v^{\ast})^{\ast}F_{+}(v^{\ast})-F_{+}(v^{\ast})F_{+}(v^{\ast})^{\ast}\in\mathcal{S}_{1}(\ell^2(\Z_{\ge 0})),\quad F_{-}(v)^{\ast}F_{-}(v)-F_{-}(v)F_{-}(v)^{\ast}\in\mathcal{S}_{1}(\ell^2(\Z_{\ge 0})).
\end{align}
\end{Lem}
\begin{proof}
We only consider $F_{+}(v^{\ast})^{\ast}F_{+}(v^{\ast}).$ The other cases can be proven similarly. By direct calculations, we have
\begin{align*}
F_{+}(v^{\ast})^{\ast}F_{+}(v^{\ast})&=\displaystyle\frac{1}{4}\left[(1+p_{+})b_{+}v-(1-p_{+})b_{+}v^{\ast}-2q_{+}a_{+}\right]\cdot \left[(1+p_{+})b_{+}v^{\ast}-(1-p_{+})b_{+}v-2q_{+}a_{+}\right]
\\
&=\displaystyle\frac{1}{4}[(1+p_{+})^{2}b_{+}^{2}vv^{\ast}-(1-p_{+}^2)b_{+}^{2}v^{2}-2(1+p_{+})b_{+}q_{+}a_{+}v-(1-p_{+}^2)b_{+}^{2}(v^{\ast})^{2}+(1-p_{+})^{2}b_{+}^{2}v^{\ast}v
\\
&\hspace{30mm}+2(1-p_{+})b_{+}q_{+}a_{+}v^{\ast}-2(1+p_{+})b_{+}q_{+}a_{+}v^{\ast}+2(1-p_{+})q_{+}a_{+}b_{+}v+4q_{+}^{2}a_{+}^{2}]
\\
&=\displaystyle\frac{1}{4}[-b_{+}^2(1-p_{+}^2)\{v^{2}+(v^{\ast})^{2}\}-4a_{+}b_{+}p_{+}q_{+}(v+v^{\ast})
\\
&\hspace{30mm}+b_{+}^{2}(1+p_{+})^{2}+(1-p_{+})^{2}b_{+}^{2}+4a_{+}^{2}q_{+}^{2}-b_{+}^{2}(1+p_{+})^{2}\Omega_{0}]
\\
&=\displaystyle\frac{1}{4}[-b_{+}^2(1-p_{+}^2)\{v^{2}+(v^{\ast})^{2}\}-4a_{+}b_{+}p_{+}q_{+}(v+v^{\ast})+2b_{+}^{2}(1+p_{+}^{2})+4a_{+}^{2}q_{+}^{2}-b_{+}^{2}(1+p_{+})^{2}\Omega_{0}],
\end{align*}
where to get the third line, we used \eqref{halfline}. 
\end{proof}

We compute the Witten indices of $F_{+}(v^{\ast})$ and $F_{-}(v)$ when they are Fredholm.
\begin{Thm}\label{fredholm}
The following assertions hold: 
\begin{enumerate}
\item The operator $F_{+}(v^{\ast})$ is Fredholm if and only if $|p_{+}|\neq |a_{+}|.$ In this case, we have
\begin{align}\label{wf+}
w(F_{+}(v^{\ast}))=W(a_+,p_+)
\end{align}
\item The operator $F_{-}(v)$ is Fredholm if and only if $|p_{-}|\neq |a_{-}|.$ In this case, we have
\begin{align}\label{wf-}
w(F_{-}(v))=-W(a_-,p_-)
\end{align}
\end{enumerate}
\end{Thm}
\begin{proof}
We only consider $F_{+}(v^{\ast})$ since the same argument is valid for $F_{-}(v).$
Note that the Witten index coincides with the Fredholm index when the operator is Fredholm. 
We follow the argument of the proof of \cite[Lemma 3.5]{MR4219958}. 
$F_{+}(v^{\ast})$ is unitarily equivalent to the Toeplitz operator $T_{F_{+}(\overline{\cdot})}$ with the symbol $F_{+}(\overline{\cdot}).$ 
For Toeplitz operator, we refer \cite[Section 3.5]{MR1074574}.
$T_{F_{+}(\overline{\cdot})}$ is Fredholm if and only if $F_{+}(e^{-\im\theta})$ ($\theta\in\R/ 2\pi\Z$) vanishes nowhere.
In this case, the Fredholm index of $T_{F_{+}(\overline{\cdot})}$ is equal to $-\mathrm{wn}(F_{+}(\overline{\cdot})),$ 
where $\mathrm{wn}(F_{+}(\overline{\cdot}))$ is the winding number of $F_{+}(e^{-\im\theta})$ ($\theta\in\R/ 2\pi\Z$) around the origin. 
By direct calculations, we have
\begin{align*}
F_{+}(e^{-\im\theta})&=\displaystyle\frac{\im}{2}[(1+p_{+})b_{+}e^{-\im\theta}-(1-p_{+})b_{+}e^{\im\theta}-2q_{+}a_{+}]
\\
&=\displaystyle\frac{\im}{2}(2p_{+}b_{+}\cos\theta-2b_{+}\im\sin\theta-2q_{+}a_{+})
\\
&=b_{+}\sin\theta+(p_{+}b_{+}\cos\theta-q_{+}a_{+})\im,\quad \theta\in\R/2\pi\Z.
\end{align*}
If $p_{+}b_{+}=0,$ then $F_{+}(e^{-\im\theta})=2b_{+}\sin\theta-q_{+}a_{+}\im.$ Thus the image of this function does not through the origin if and only if $q_{+}a_{+}\neq p_{+}b_{+}(=0).$ In this case, $\mathrm{wn}(F_{+}(\overline{\cdot}))=0.$ We suppose that $p_{+}b_{+}\neq 0.$ Then the image of $F_{+}(e^{-\im\theta})$ $(\theta\in\R/2\pi \Z)$ represents the ellipse centered at $-q_{+}a_{+}\im.$ This ellipse does not through the origin if and only if $|p_{+}b_{+}|\neq |q_{+}a_{+}|.$ This relation is also equivalent to $|a_{+}|\neq |p_{+}|.$ Thus, we have
\begin{align*}
\mathrm{wn}(F_{+}(\overline{\cdot}))
=
\begin{cases}
-\sgn(p_{+}), & |p_{+}|>|a_{+}|, \\
0, &|p_{+}|<|a_{+}|.
\end{cases}
\end{align*}
A direct computation shows that the right hand side coincides with $W(a_+,p_+)$.
Hence, \eqref{wf+} follows. 
\end{proof}
In what  follows, we assume both $|a_{\pm}|=|p_{\pm}|.$ In this case, $F_{+}(v^{\ast})$ and $F_{-}(v)$ are not Fredholm.
We first derive $w(F_{+}(v^*))$ and $w(F_{-}(v))$ under exceptional cases $p_{+}=1,$ $p_{+}=-1,$ $p_{-}=1$ or $p_{-}=-1.$

\begin{Thm}
\label{exceptional} 
The following assertions hold:
\begin{enumerate}
\item If $p_{+}=1$ or $p_{+}= -1,$ then $w(F_{+}(v^*))=0.$
\item If $p_{-}=1$ or $p_{-}=-1,$ then $w(F_{-}(v))=0.$
\end{enumerate}
\end{Thm}
\begin{proof}
If $p_{+}=\pm 1$ (resp.~$p_{-}=\pm1$), then $b_{+}=q_{+}=0$ (resp.~$b_{-}=q_{-}=0$). Thus, $F_{+}(v^{\ast})=0$ (resp.~$F_{-}(v)=0$). From this, we have $w(F_{+}(v^{\ast}))=0$ (resp.~$w(F_{-}(v))=0$).
\end{proof}
The rest of this paper is devoted to compute $w(F_{+}(v^*))$ and $w(F_{-}(v))$ under the conditions $|p_{\pm}|=|a_{\pm}|<1.$ 
We notice that in this case we have $b_{\pm}=q_{\pm}.$ 
We rewrite $b_\pm$ and $q_\pm$ in the functions $F_\pm$ using $a_\pm$ and $p_\pm.$
\begin{Lem}
\label{gastg}
We have
\begin{align*}
F_{+}(v^{\ast})^{\ast}F_{+}(v^{\ast})&=\displaystyle\frac{1-p_{+}^2}{4}\Big[-(1-p_{+}^2)\{v^2+(v^{\ast})^2\}-4a_{+}p_{+}(v+v^{\ast})+2(1+3p_{+}^2)-(1+p_{+})^{2}\Omega_{0}\Big],
\\
F_{+}(v^{\ast})F_{+}(v^{\ast})^{\ast}&=\displaystyle\frac{1-p_{+}^2}{4}\Big[-(1-p_{+}^2)\{v^2+(v^{\ast})^2\}-4a_{+}p_{+}(v+v^{\ast})+2(1+3p_{+}^2)-(1-p_{+})^{2}\Omega_{0}\Big],
\\
F_{-}(v)^{\ast}F_{-}(v)&=\displaystyle\frac{1-p_{-}^2}{4}\Big[-(1-p_{-}^2)\{v^2+(v^{\ast})^2\}-4a_{-}p_{-}(v+v^{\ast})+2(1+3p_{-}^2)-(1-p_{-})^{2}\Omega_{0}\Big],
\\
F_{-}(v)F_{-}(v)^{\ast}&=\displaystyle\frac{1-p_{-}^2}{4}\Big[-(1-p_{-}^2)\{v^2+(v^{\ast})^2\}-4a_{-}p_{-}(v+v^{\ast})+2(1+3p_{-}^2)-(1+p_{-})^{2}\Omega_{0}\Big].
\end{align*}
\end{Lem}
\begin{proof}
In this proof, we only consider $F_{+}(v^{\ast})^{\ast}F_{+}(v^{\ast}).$ The other cases can be proven similarly. 
We recall the relations $|a_{+}|=|p_{+}|$ and $b_{+}=q_{+}$, and thus $b_{+}^2=1-p_{+}^2$ follows. 
By Lemma \ref{lem:fastf}, we have
\begin{align*}
\ F_{+}(v^{\ast})^{\ast}F_{+}(v^{\ast})&=\displaystyle\frac{1}{4}[-b_{+}^2(1-p_{+}^2)\{v^{2}+(v^{\ast})^{2}\}-4a_{+}b_{+}p_{+}q_{+}(v+v^{\ast})+2b_{+}^{2}(1+p_{+}^{2})+4a_{+}^{2}q_{+}^{2}-b_{+}^{2}(1+p_{+})^{2}\Omega_{0}]
\\
&=\displaystyle\frac{b_{+}^{2}}{4}[-(1-p_{+}^2)\{v^{2}+(v^{\ast})^{2}\}-4a_{+}p_{+}(v+v^{\ast})+2(1+p_{+}^2)+4a_{+}^2-(1+p_{+})^{2}\Omega_{0}]
\\
&=\displaystyle\frac{1-p_{+}^2}{4}\Big[-(1-p_{+}^2)\{v^2+(v^{\ast})^2\}-4a_{+}p_{+}(v+v^{\ast})+2(1+3p_{+}^2)-(1+p_{+})^{2}\Omega_{0}\Big].
\end{align*}
This completes the proof.
\end{proof}
The following equalities are obtained by expanding the right-hand sides and by using Lemma \ref{gastg}.
The proof is not difficult and we omit it.
\begin{Lem}\label{square}
We have
\begin{align*}
F_{+}(v^{\ast})^{\ast}F_{+}(v^{\ast})&=-\displaystyle\frac{(1-p_{+}^2)^2}{4}\left[\left\{(v+v^{\ast})+\displaystyle\frac{2a_{+}p_{+}}{1-p_{+}^2}\right\}^2+\displaystyle\frac{2}{1-p_{+}}\Omega_{0}\right]+1,
\\
F_{+}(v^{\ast})F_{+}(v^{\ast})^{\ast}&=-\displaystyle\frac{(1-p_{+}^2)^2}{4}\left[\left\{(v+v^{\ast})+\displaystyle\frac{2a_{+}p_{+}}{1-p_{+}^2}\right\}^2+\displaystyle\frac{2}{1+p_{+}}\Omega_{0}\right]+1,
\\
F_{-}(v)^{\ast}F_{-}(v)&=-\displaystyle\frac{(1-p_{-}^2)^2}{4}\left[\left\{(v+v^{\ast})+\displaystyle\frac{2a_{-}p_{-}}{1-p_{-}^2}\right\}^2+\displaystyle\frac{2}{1+p_{-}}\Omega_{0}\right]+1,
\\
F_{-}(v)F_{-}(v)^{\ast}&=-\displaystyle\frac{(1-p_{-}^2)^2}{4}\left[\left\{(v+v^{\ast})+\displaystyle\frac{2a_{-}p_{-}}{1-p_{-}^2}\right\}^2+\displaystyle\frac{2}{1-p_{-}}\Omega_{0}\right]+1.
\end{align*}
\end{Lem}

\section{Perturbation determinant}
\label{section: section 4}
We introduce the following two operators on $\ell^2(\Z_{\ge 0}):$
\begin{align}
T(A, P):=\left\{(v+v^{\ast})+\displaystyle\frac{2AP}{1-P^2}\right\}^{2}+\displaystyle\frac{2}{1-P}\Omega_{0},\quad T_{0}(A, P):=\left\{(v+v^{\ast})+\displaystyle\frac{2AP}{1-P^2}\right\}^{2},
\end{align}
where $A$ and $P$ are real numbers satisfying $|A|<1,$ $|P|< 1$ and $|A|=|P|.$ Note that both $T(A, P)$ and $T_{0}(A, P)$ are non-negative self-adjoint. We set a positive constant $\alpha_{P}$ by
\begin{align}\label{alpha-p}
\alpha_{P}:=\displaystyle\frac{(1-P^2)^2}{4},\quad -1<P<1.
\end{align}
We can represent $F_{+}(v^{\ast})^{\ast}F_{+}(v^{\ast}),$ $F_{+}(v^{\ast})F_{+}(v^{\ast})^{\ast},$ $F_{-}(v)^{\ast}F_{-}(v)$ and $F_{-}(v)F_{-}(v)^{\ast}$ using $T(A, P).$ 
More precisely, see Table 1 below.
\begin{table}[htb]
\begin{center}
  \begin{tabular}{|c|c|c|} \hline
    $A$ & $P$ & $T(A, P)$ \\ \hline \hline
    $a_{+}$ & $p_{+}$  & $-(F_{+}(v^{\ast})^{\ast}F_{+}(v^{\ast})-1)/\alpha_{p_{+}}$ \\
    $-a_{+}$ & $-p_{+}$ & $-(F_{+}(v^{\ast})F_{+}(v^{\ast})^{\ast}-1)/\alpha_{-p_{+}}$ \\
    $-a_{-}$ & $-p_{-}$  & $-(F_{-}(v)^{\ast}F_{-}(v)-1)/\alpha_{-p_{-}}$ \\ 
    $a_{-}$ & $p_{-}$  & $-(F_{-}(v)F_{-}(v)^{\ast}-1)/\alpha_{p_{-}}$ \\ \hline
  \end{tabular}
\end{center}
\end{table}
\begin{center}
Table 1: Correspondence between $T(A, P)$ and $F_\pm$'s. For example, if we set $A=a_{+}$ and $P=p_{+},$ then we get $T(a_{+}, p_{+})=-(F_{+}(v^{\ast})^{\ast}F_{+}(v^{\ast})-1)/\alpha_{p_{+}}.$
\end{center}
Now $\ind_{t}(Q_{\epsilon_{0}})$ is modified as
\begin{align*}
\ind_{t}(Q_{\epsilon_{0}})&=\Tr\left[e^{-tF_{+}(v^{\ast})^{\ast}F_{+}(v^{\ast})}-e^{-tF_{+}(v^{\ast})F_{+}(v^{\ast})^*}\right]+\Tr\left[e^{-tF_{-}(v)^{\ast}F_{-}(v)}-e^{-tF_{-}(v)F_{-}(v)^{\ast}}\right]
\\
&=e^{-t}\Tr\left[ e^{t\alpha_{p_{+}}T(a_{+}, p_{+})}-e^{t\alpha_{-p_{+}}T(-a_{+}, -p_{+})}\right]+e^{-t}\Tr \left[e^{t\alpha_{-p_{-}}T(-a_{-}, -p_{-})}-e^{t\alpha_{p_{-}}T(a_{-}, p_{-})}\right]
\\
&=e^{-t}\Tr\left[e^{t\alpha_{p_{+}}T(a_{+}, p_{+})}-e^{t\alpha_{p_{+}}T_{0}(a_{+}, p_{+})}\right]+e^{-t}\Tr\left[e^{t\alpha_{-p_{+}}T(-a_{+}, -p_{+})}-e^{t\alpha_{-p_{+}}T_{0}(-a_{+}, -p_{+})}\right]
\\
&\hspace{5mm}+e^{-t}\Tr\left[e^{t\alpha_{-p_{-}}T(-a_{-}, -p_{-})}-e^{t\alpha_{-p_{+}}T_{0}(-a_{-}, -p_{-})}\right]+e^{-t}\Tr\left[e^{t\alpha_{p_{-}}T(a_{-}, p_{-})}-e^{t\alpha_{p_{-}}T_{0}(a_{-}, p_{-})}\right],
\end{align*}
where we used $\alpha_{P}=\alpha_{-P}$ and $T_{0}(A, P)=T_{0}(-A, -P).$ From the above equality, it suffices to consider $\Tr(e^{t\alpha_{P} T(A, P)}-e^{t\alpha_{P}T_{0}(A, P)}).$
\begin{Rem}\label{symmetry}
We define the unitary operator $W$ on $\ell^2(\Z_{\ge 0})$ by
$$ (W\psi)(x):=(-1)^x\psi(x),\quad x\in \Z_{\ge 0},\quad \psi\in \ell^2(\Z_{\ge 0}).$$
Since $WvW^{-1}=-v,$ we have $WT(A, P)W^{-1}=T(-A, P)$ and $WT_{0}(A, P)W^{-1}=T_{0}(-A, P).$ This implies that 
$$ \Tr(e^{t\alpha_{P} T(A, P)}-e^{t\alpha_{P}T_{0}(A, P)})=\Tr(e^{t\alpha_{P} T(-A, P)}-e^{t\alpha_{P}T_{0}(-A, P)}). $$
Thus , $\Tr(e^{t\alpha_{P} T(A, P)}-e^{t\alpha_{P}T_{0}(A, P)})$ does not depend on the choice of $A=P$ or $A=-P.$
\end{Rem}
By Remark \ref{symmetry}, without loss of generality, we choose $A$ as $A=-P$ for convenience. We denote $T(-P, P)$ as $T(P)$ and $T_{0}(-P, P)$ as $T_{0}(P),$ respectively:
\begin{align}\label{def:T}
T(P)&:=T(-P, P)=\left\{(v+v^{\ast})-\displaystyle\frac{2P^2}{1-P^2}\right\}^{2}+\displaystyle\frac{2}{1-P}\Omega_{0},
\\ \label{def:T0}
T_{0}(P)&:=T_{0}(-P, P)=\left\{(v+v^{\ast})-\displaystyle\frac{2P^2}{1-P^2}\right\}^{2}.
\end{align}
To calculate $\Tr(e^{t\alpha_{P} T(P)}-e^{t\alpha_{P}T_{0}(P)}),$ we make use of the spectral shift function defined by
\begin{align}\label{ssf}
\xi(x)&:=\displaystyle\frac{1}{\pi}\displaystyle\lim_{\e\rightarrow +0}\mathrm{Arg}\ \Delta_{T(P)/T_{0}(P)}(x+\im\epsilon), \quad x\in\R,
\end{align}
where $\mathrm{Arg}$ is the argument such that $\mathrm{Arg}(z)\in(-\pi, \pi]$ for any $z\in\C\setminus\{0\}$ and $\Delta_{T(P)/T_{0}(P)}(\cdot)$ is the perturbation determinant for the pair $(T(P), T_{0}(P))$;
\begin{align}\label{def:pd}
\Delta_{T(P)/T_{0}(P)}(z)&:=\det((T(P)-z)(T_{0}(P)-z)^{-1})=1+\displaystyle\frac{2}{1-P}\langle \delta_{0}, (T_{0}(P)-z)^{-1}\delta_{0}\rangle,\quad z\in\mathbb{C}\setminus\R.
\end{align}

In what follows, we simply write $T(P)$ and $T_{0}(P)$ as $T$ and $T_{0},$ respectively if there is no danger of confusion.

In the rest of this paper, a square root of a complex number appears. 
We only treat the principle value of a square root. Namely, for any $z\in\C\setminus(-\infty, 0],$ we define
\begin{align}\label{root}
\sqrt{z}:=r^{1/2}e^{\frac{\im}{2}\mathrm{Arg}(z)},\quad z=re^{\im\mathrm{Arg}(z)}, \quad r>0, \quad \mathrm{Arg}(z)\in(-\pi, \pi).
\end{align}
We introduce a constant $m_{P}$ by
\begin{align}\label{mp}
m_{P}:=\displaystyle\frac{P^2}{1-P^2},\quad -1<P<1. 
\end{align}
We note that 
\begin{align*}
0\le m_{P}<1\quad &\mbox{ if and only if }\quad 0\le |P|<1/\sqrt{2}, \\
m_{P}=1\quad  &\mbox{ if and only if }\quad |P|=1/\sqrt{2}, \\ 
m_{P}>1\quad  &\mbox{ if and only if }\quad 1/\sqrt{2}<|P|<1.
\end{align*}
Moreover, it is seen that 
\begin{align}\label{alphap-mp}
\alpha_{P}^{-1}=4(1+m_{P})^2.
\end{align}
Although the following lemma is simple, it is a key ingredients in this paper because $\Delta_{T/T_{0}}(z)$ 
can be represented as the difference of the resolvents of the discrete Laplace operators on the half-line. 
\begin{Lem}\label{resolvent}
For any $z\in \C\setminus \R,$ we have
$$ \Delta_{T/T_{0}}(z)=1+\displaystyle\frac{H(\tau_{+}(z))-H(\tau_{-}(z))}{(1-P)\sqrt{z}}, $$
 where 
$$ \tau_{\pm}(z):=2m_{P}\pm\sqrt{z},\quad z\in\C\setminus(-\infty,0],$$
and
$$\quad H(z):=\displaystyle\frac{\sqrt{\frac{z-2}{z+2}}-1}{\sqrt{\frac{z-2}{z+2}}+1},\quad z\in\C\setminus[-2, 2].$$
\end{Lem}
\begin{proof}
By the spectral decomposition of $v+v^{\ast}$ (see  Appendix \ref{section: Appendix B}), 
it follows that
\begin{equation*}
\begin{aligned}
\Delta_{T/T_{0}}(z)&=1+\displaystyle\frac{2}{1-P}\langle \delta_{0}, (T_{0}-z)^{-1}\delta_{0}\rangle
\\
&=1+\displaystyle\frac{2}{1-P}\displaystyle\int_{-2}^{2}\frac{\sqrt{4-t^{2}}}{2\pi}\frac{1}{\left(t-2m_{P}\right)^2-z}\mathrm{d}t
\vspace{3mm}
\\
&=1+\displaystyle\frac{1}{(1-P)\sqrt{z}}\displaystyle\int_{-2}^{2}\frac{\sqrt{4-t^{2}}}{2\pi}\left[\frac{1}{t-2m_{P}-\sqrt{z}}-\frac{1}{t-2m_{P}+\sqrt{z}}\right]\mathrm{d}t
\\
&=1+\displaystyle\frac{H(\tau_{+}(z))-H(\tau_{-}(z))}{(1-P)\sqrt{z}},
\end{aligned}
\end{equation*}
where to obtain the last equality, we used the fifth assertion of Proposition \ref{theorem: properties of H}.
\end{proof}

\begin{Rem}\label{xi:positive}
Since $T_{0}$ is non-negative, 
the resolvent set of $T_{0}$ contains $\C\setminus[0,\infty).$
Thus we have
$$ \displaystyle\lim_{\epsilon\rightarrow+0}\Delta_{T/T_{0}}(x+\im\epsilon)=1+\displaystyle\frac{\langle\delta_{0}, (T_{0}-x)^{-1}\delta_{0}\rangle}{1-P}>0,\quad x<0.$$
This implies that $\xi(x)=0$ for any $x<0.$ Moreover, we will see in Lemma \ref{xi:positivezero} that $\xi(x)=0$ for any $x>4(1+m_{P})^2$.
\end{Rem}
We shall see in Lemma \ref{lem:pdlimit} that the limit $H(\tau_{\pm}(x+\im\epsilon))$ as $\epsilon\rightarrow+0$ takes one of three possible values. 
Before going to Lemma \ref{lem:pdlimit}, we need a lemma.
\begin{Lem}\label{lem:taupm}
For any $-1<P<1$ and $x>0,$ the following assertions hold:
\begin{enumerate}
\item If $0\le|P|<1/\sqrt{2},$ then
\begin{equation*}
\begin{cases}
|\tau_{+}(x)|<2 \mbox{ and } |\tau_{-}(x)|<2, &0<x<4(1-m_{P})^2, \\
\tau_{+}(x)>2 \mbox{ and } |\tau_{-}(x)|<2, &4(1-m_{P})^2<x<4(1+m_{P})^2, \\
\tau_{+}(x)>2 \mbox{ and } \tau_{-}(x)<-2, &4(1+m_{P})^{2}<x.
\end{cases}
\end{equation*}
\item If $|P|=1/\sqrt{2},$ then
\begin{equation*}
\begin{cases}
\tau_{+}(x)>2 \mbox{ and } |\tau_{-}(x)|<2, &0<x<16, \\
\tau_{+}(x)>2 \mbox{ and } \tau_{-}(x)<-2, & 16<x.
\end{cases}
\end{equation*}
\item If $1/\sqrt{2}<|P|<1,$ then
\begin{equation*}
\begin{cases}
\tau_{+}(x)>2 \mbox{ and } \tau_{-}(x)>2, & 0<x<4(m_{P}-1)^2, \\
\tau_{+}(x)>2 \mbox{ and } |\tau_{-}(x)|<2, &4(m_{P}-1)^2<x<4(m_{P}+1)^2, \\
\tau_{+}(x)>2 \mbox{ and } \tau_{-}(x)<-2, &4(m_{P}+1)^{2}<x.
\end{cases}
\end{equation*}
\end{enumerate}
\end{Lem}

\begin{figure}[H]
\centering
\begin{tikzpicture}[every node/.style={scale = 0.8}, scale = 0.8]
\begin{axis}[
width = 0.9\textwidth, height = 0.5\textwidth,
xmin=-1.5, xmax=5, 
ymin= -2, ymax=2, 
legend pos = north west, 
axis lines=center, 
axis y line=none,
xlabel=$x$, ylabel=$y$, xlabel style = {anchor = west}, xtick= {0, 2, 4}, xticklabels=\empty, clip=false
]
\node at (-1.5, 1)[right]  {1. $0 \leq |P| < \frac{1}{\sqrt{2}}$};

\node at (0,0) [below] {$0$};
\node at (2,0) [below] {$4(1 - m_p)^2$};
\node at (4,0) [below] {$4(1 + m_p)^2$};

\draw[latex'-latex'] (0,-1) -- (4,-1);
\draw[densely dashed, latex'-latex'] (4,-1) -- (5,-1);
\node at (2,-1) [below] {$|\tau_-| < 2$};
\node at (4.5,-1) [below] {$\tau_- < -2$};

\draw[latex'-latex'] (0,1) -- (2,1);
\draw[densely dashed, latex'-latex'] (2,1) -- (5,1);
\node at (1,1) [above] {$|\tau_+| < 2$};
\node at (3.5,1) [above] {$\tau_+ > 2$};

\end{axis}
\begin{axis}[
yshift= -50mm,
width = 0.9\textwidth, height = 0.5\textwidth,
xmin=-1.5, xmax=5, 
ymin= -2, ymax=2, 
legend pos = north west, 
axis lines=center, 
axis y line=none,
xlabel=$x$, ylabel=$y$, xlabel style = {anchor = west}, xtick= {0, 4}, xticklabels=\empty, clip=false
]
\node at (-1.5, 1) [right]{2. $|P| = \frac{1}{\sqrt{2}}$};

\node at (0,0) [below] {$0$};
\node at (4,0) [below] {$16$};

\draw[densely dashed, latex'-latex'] (0,-1) -- (2,-1);
\draw[latex'-latex'] (2,-1) -- (4,-1);
\draw[densely dashed, latex'-latex'] (4,-1) -- (5,-1);

\draw[latex'-latex'] (0,-1) -- (4,-1);
\draw[densely dashed, latex'-latex'] (4,-1) -- (5,-1);
\node at (2,-1) [below] {$|\tau_-| < 2$};
\node at (4.5,-1) [below] {$\tau_- < -2$};

\draw[densely dashed, latex'-latex'] (0,1) -- (5,1);
\node at (2.5,1) [above] {$\tau_+ > 2$};
\end{axis}

\begin{axis}[
yshift= - 100mm,
width = 0.9\textwidth, height = 0.5\textwidth,
xmin=-1.5, xmax=5, 
ymin= -2, ymax=2, 
legend pos = north west, 
axis lines=center, 
axis y line=none,
xlabel=$x$, ylabel=$y$, xlabel style = {anchor = west}, xtick= {0, 2, 4}, xticklabels=\empty, clip=false
]
\node at (-1.5, 1) [right] {3. $\frac{1}{\sqrt{2}} \leq |P| < 1$};

\node at (0,0) [below] {$0$};
\node at (2,0) [below] {$4(1 - m_p)^2$};
\node at (4,0) [below] {$4(1 + m_p)^2$};

\draw[densely dashed, latex'-latex'] (0,-1) -- (2,-1);
\draw[latex'-latex'] (2,-1) -- (4,-1);
\draw[densely dashed, latex'-latex'] (4,-1) -- (5,-1);

\node at (1,-1) [below] {$\tau_- > 2$};
\node at (3,-1) [below] {$|\tau_-| < -2$};
\node at (4.5,-1) [below] {$\tau_- < -2$};

\draw[densely dashed, latex'-latex'] (0,1) -- (5,1);
\node at (2.5,1) [above] {$\tau_+ > 2$};

\end{axis}
\end{tikzpicture}
\caption{Values of $\tau_{\pm}.$}
\label{figure: graph of Lambda}
\end{figure}

\begin{proof}
We consider the following six cases:
\begin{enumerate}
\item $\tau_{+}(x)>2.$
\begin{equation*}
2m_{P}+\sqrt{x}>2\quad \mbox{ if and only if }\quad \sqrt{x}>2-2m_{P}.
\end{equation*}
The solution of the above inequality is
\begin{align*}
\begin{cases}x>4(1-m_{P})^2, \quad  &0\le |P|<\frac{1}{\sqrt{2}},
\\
x>0, & |P|=\frac{1}{\sqrt{2}}, \\
x>0,& \frac{1}{\sqrt{2}}<|P|<1.
\end{cases}
\end{align*}
\item $\tau_{+}(x)<-2.$
\begin{equation*}
2m_{P}+\sqrt{x}<-2\quad \mbox{ if and only if }\quad \sqrt{x}<-2-2m_{P}.
\end{equation*}
The solution of the above inequality is the empty set for all $P.$
\item $|\tau_{+}(x)|<2.$
\begin{equation*}
-2<2m_{P}+\sqrt{x}<2\quad \mbox{ if and only if }\quad -2-2m_{P}<\sqrt{x}<2-2m_{P}.
\end{equation*}
The solution of the above inequality is
\begin{align*}
\begin{cases}x<4(1-m_{P})^2, \quad  &0\le |P|<\frac{1}{\sqrt{2}},
\\
\emptyset, & |P|=\frac{1}{\sqrt{2}}, \\
\emptyset, & \frac{1}{\sqrt{2}}<|P|<1.
\end{cases}
\end{align*}
\item $\tau_{-}(x)>2.$
\begin{equation*}
2m_{P}-\sqrt{x}>2\quad \mbox{ if and only if }\quad 2m_{P}-2>\sqrt{x}.
\end{equation*}
The solution of the above inequality is
\begin{align*}
\begin{cases}\emptyset, \quad  &0\le |P|<\frac{1}{\sqrt{2}},
\\
\emptyset, & |P|=\frac{1}{\sqrt{2}}, \\
0<x<4(m_{P}-1)^2, & \frac{1}{\sqrt{2}}<|P|<1.
\end{cases}
\end{align*}

\item $\tau_{-}(x)<-2.$
\begin{equation*}
2m_{P}-\sqrt{x}<-2\quad \mbox{ if and only if }\quad 2m_{P}+2<\sqrt{x}.
\end{equation*}
The solution of the above inequality is
\begin{align*}
\begin{cases}4(1+m_{P})^2<x, \quad  &0\le |P|<\frac{1}{\sqrt{2}},
\\
16<x, & |P|=\frac{1}{\sqrt{2}}, \\
4(1+m_{P})^2<x, & \frac{1}{\sqrt{2}}<|P|<1.
\end{cases}
\end{align*}
\item $|\tau_{-}(x)|<2.$
\begin{equation*}
-2<2m_{P}-\sqrt{x}<2\quad \mbox{ if and only if }\quad -2+2m_{P}<\sqrt{x}<2+2m_{P}.
\end{equation*}
The solution of the above inequality is
\begin{align*}
\begin{cases}x<4(1+m_{P})^2, \quad  &0\le |P|<\frac{1}{\sqrt{2}},
\\
0<x<16, & |P|=\frac{1}{\sqrt{2}}, \\
4(m_{P}-1)^2<x< 4(1+m_{P})^2,& \frac{1}{\sqrt{2}}<|P|<1.
\end{cases}
\end{align*}
\end{enumerate}
By summarizing the above six cases, we get the desired result.
\end{proof}


\begin{Lem}\label{lem:pdlimit} 
For any $x>0,$ we have
\begin{equation*}
\displaystyle\lim_{\epsilon\rightarrow+0}\{\tau_{\pm}(x+\im\epsilon)+2\}\sqrt{\displaystyle\frac{\tau_{\pm}(x+\im\epsilon)-2}{\tau_{\pm}(x+\im\epsilon)+2}}=\begin{cases}\sqrt{\tau_{\pm}(x)^2-4},\hspace{5mm}& \tau_{\pm}(x)>2, \\ -\sqrt{\tau_{\pm}(x)^2-4},\hspace{5mm}&  \tau_{\pm}(x)<-2, \\ \pm \im\sqrt{4-\tau_{\pm}(x)^2},\hspace{5mm}& |\tau_{\pm}(x)|<2.\end{cases}
\end{equation*}
\end{Lem}
\begin{proof}
\leavevmode
\begin{enumerate}
\item $\tau_{\pm}(x)>2.$
\\
Since $(\tau_{\pm}(x)-2)/(\tau_{\pm}(x)+2)>0$ and $\sqrt{z}$ is holomorphic on $\mathbb{C}\setminus(-\infty, 0],$ we have
\begin{align*}
\displaystyle\lim_{\epsilon\rightarrow+0}\{\tau_{\pm}(x+\im\epsilon)+2\}\sqrt{\displaystyle\frac{\tau_{\pm}(x+\im\epsilon)-2}{\tau_{\pm}(x+\im\epsilon)+2}}&=(\tau_{\pm}(x)+2)\sqrt{\displaystyle\frac{\tau_{\pm}(x)-2}{\tau_{\pm}(x)+2}}=\sqrt{(\tau_{\pm}(x)+2)^2\displaystyle\frac{\tau_{\pm}(x)-2}{\tau_{\pm}(x)+2}}=\sqrt{\tau_{\pm}(x)^2-4}.
\end{align*}
\\
\item  $\tau_{\pm}(x)<-2.$
\\
Similarly, from $(\tau_{\pm}(x)-2)/(\tau_{\pm}(x)+2)>0$ and $-\tau_{\pm}(x)-2>0,$ we have
\begin{align*}
\displaystyle\lim_{\epsilon\rightarrow+0}\{\tau_{\pm}(x+\im\epsilon)+2\}\sqrt{\displaystyle\frac{\tau_{\pm}(x+\im\epsilon)-2}{\tau_{\pm}(x+\im\epsilon)+2}}=-(-\tau_{\pm}(x)-2)\sqrt{\displaystyle\frac{\tau_{\pm}(x)-2}{\tau_{\pm}(x)+2}}&=-\sqrt{(-\tau_{\pm}(x)-2)^2\displaystyle\frac{\tau_{\pm}(x)-2}{\tau_{\pm}(x)+2}}
\\
&=-\sqrt{\tau_{\pm}(x)^2-4}.
\end{align*}
\item  $|\tau_{\pm}(x)|<2.$
\\
It is seen that $(\tau_{\pm}(x)-2)/(\tau_{\pm}(x)+2)<0.$ We set $\C_{\pm}$ as 
$$ \C_{\pm}:=\{z\in\C|\ \pm \mathrm{Im}\,z>0\}.$$
Since $\tau_{\pm}(x+\im\epsilon)\in\mathbb{C}_{\pm},$ �we get
$$\sqrt{(\tau_{\pm}(x+\im\epsilon)-2)/(\tau_{\pm}(x+\im\epsilon)+2)}\in\mathbb{C}_{\pm}.$$
Thus, we have
\begin{align*}
\displaystyle\lim_{\epsilon\rightarrow+0}\{\tau_{\pm}(x+\im\epsilon)+2\}\sqrt{\displaystyle\frac{\tau_{\pm}(x+\im\epsilon)-2}{\tau_{\pm}(x+\im\epsilon)+2}}=\{\tau_{\pm}(x)+2\}\sqrt{\left|\displaystyle\frac{\tau_{\pm}(x)-2}{\tau_{\pm}(x)+2}\right|}e^{\pm\im\frac{\pi}{2}}&=\pm \im\sqrt{(\tau_{\pm}(x)+2)^2\displaystyle\frac{2-\tau_{\pm}(x)}{\tau_{\pm}(x)+2}}
\\
&=\pm \im\sqrt{4-\tau_{\pm}(x)^2}.
\end{align*}
\end{enumerate}
\end{proof}

\begin{Lem}
\label{det:limzero}
For any $x>0,$ the following assertions hold:
\leavevmode
\begin{enumerate}
\item If $0\le|P|<\displaystyle\frac{1}{\sqrt{2}},$ then
\begin{align*}
\displaystyle\lim_{\epsilon\rightarrow +0}&\Delta_{T/T_{0}}(x+\im\epsilon) \\
&=
\begin{cases}
\displaystyle \frac{-2P\sqrt{x}+\im\sqrt{4-\left(2m_{P}+\sqrt{x}\right)^2}+\im\sqrt{4-\left(2m_{P}-\sqrt{x}\right)^2}}{2(1-P)\sqrt{x}}, &0<x<4(1-m_{P})^2, \\[10pt]
\displaystyle \frac{-2P\sqrt{x}+\sqrt{\left(2m_{P}+\sqrt{x}\right)^2-4}+\im\sqrt{4-\left(2m_{P}-\sqrt{x}\right)^2}}{2(1-P)\sqrt{x}}, &4(1-m_{P})^2<x<4(1+m_{P})^2.
\end{cases}
\end{align*}
\item If $|P|=\displaystyle\frac{1}{\sqrt{2}},$ then
\begin{align*}
\lim_{\epsilon\rightarrow+0}\Delta_{T/T_{0}}(x+\im\epsilon) =
\frac{-2P\sqrt{x}+\sqrt{4\sqrt{x}+x}+\im\sqrt{4\sqrt{x}-x}}{2(1-P)\sqrt{x}}, \qquad 0<x<16.
\end{align*}
\item If $\displaystyle\frac{1}{\sqrt{2}}<|P|<1,$ then
\begin{align*}
\displaystyle\lim_{\epsilon\rightarrow +0}&\Delta_{T/T_{0}}(x+\im\epsilon) \\
&=
\begin{cases}
\displaystyle\frac{-2P\sqrt{x}+\sqrt{\left(2m_{P}+\sqrt{x}\right)^2-4}-\sqrt{\left(2m_{P}-\sqrt{x}\right)^2-4}+\im 0}{2(1-P)\sqrt{x}}, & 0<x<4(m_{P}-1)^2, \\[10pt]
\displaystyle\frac{-2P\sqrt{x}+\sqrt{\left(2m_{P}+\sqrt{x}\right)^2-4}+\im\sqrt{4-\left(2m_{P}-\sqrt{x}\right)^2}}{2(1-P)\sqrt{x}}, & 4(m_{P}-1)^2<x<4(m_{P}+1)^2.
\end{cases}
\end{align*}
\end{enumerate}
\end{Lem}
\begin{proof}
We use the following expression (see the proof of Proposition \ref{theorem: properties of H}):
\begin{equation*}
H(z)=\displaystyle\frac{-z+(z+2)\sqrt{\displaystyle\frac{z-2}{z+2}}}{2},\hspace{5mm}z\in\mathbb{C}\setminus[-2, 2].
\end{equation*}
Then, we have 
\begin{equation*}
\begin{aligned}
&\hspace{5mm}\displaystyle\lim_{\epsilon\rightarrow +0}\Delta_{T/T_{0}}(x+\im\epsilon)
\\
&=\displaystyle\lim_{\epsilon\rightarrow +0}\left[1+\displaystyle\frac{H(\tau_{+}(x+\im\epsilon))-H(\tau_{-}(x+\im\epsilon))}{(1-P)\sqrt{x+\im\epsilon}}\right]
\\
&=1+\displaystyle\frac{1}{2(1-P)\sqrt{x}}\displaystyle\lim_{\epsilon\rightarrow +0}\Big[-\tau_{+}(x+\im\epsilon)+(\tau_{+}(x+\im\epsilon)+2)\sqrt{\displaystyle\frac{\tau_{+}(x+\im\epsilon)-2}{\tau_{+}(x+\im\epsilon)+2}}
\\
&\hspace{80mm}+\tau_{-}(x+\im\epsilon)-(\tau_{-}(x+\im\epsilon)+2)\sqrt{\displaystyle\frac{\tau_{-}(x+\im\epsilon)-2}{\tau_{-}(x+\im\epsilon)+2}}\Big]
\\
&=\displaystyle\frac{1}{2(1-P)\sqrt{x}}\displaystyle\lim_{\epsilon\rightarrow +0}\left[-2P\sqrt{x}+(\tau_{+}(x+\im\epsilon)+2)\sqrt{\displaystyle\frac{\tau_{+}(x+\im\epsilon)-2}{\tau_{+}(x+\im\epsilon)+2}}-(\tau_{-}(x+\im\epsilon)+2)\sqrt{\displaystyle\frac{\tau_{-}(x+\im\epsilon)-2}{\tau_{-}(x+\im\epsilon)+2}}\right].
\end{aligned}
\end{equation*}
By applying Lemma \ref{lem:taupm} and Lemma \ref{lem:pdlimit}, we have the following:
\begin{enumerate}
\item If $0\le|P|<\displaystyle\frac{1}{\sqrt{2}},$ then
\begin{equation*}
\begin{aligned}
&\hspace{5mm}\displaystyle\lim_{\epsilon\rightarrow +0}\Delta_{T/T_{0}}(x+\im\epsilon)
\\
&=\begin{cases}\displaystyle\frac{1}{2(1-P)\sqrt{x}}\left[-2P\sqrt{x}+\im\sqrt{4-\tau_{+}(x)^2}+\im\sqrt{4-\tau_{-}(x)^2}\right],\quad 0<x<4(1-m_{P})^2,
\vspace{3mm}
\\
\displaystyle\frac{1}{2(1-P)\sqrt{x}}\left[-2P\sqrt{x}+\sqrt{\tau_{+}(x)^2-4}+\im\sqrt{4-\tau_{-}(x)^2}\right],\quad 4(1-m_{P})^2<x<4(1+m_{P})^2.
\end{cases}
\end{aligned}
\end{equation*}
\item If $|P|=\displaystyle\frac{1}{\sqrt{2}},$ then
\begin{equation*}
\begin{aligned}
&\hspace{5mm}\displaystyle\lim_{\epsilon\rightarrow +0}\Delta_{T/T_{0}}(x+\im\epsilon)
\\
&=\displaystyle\frac{1}{2(1-P)\sqrt{x}}\left[-2P\sqrt{x}+\sqrt{\tau_{+}(x)^2-4}+\im\sqrt{4-\tau_{-}(x)^2}\right],\quad 0<x<16.
\end{aligned}
\end{equation*}
\item If $\displaystyle\frac{1}{\sqrt{2}}<|P|<1,$ then
\begin{equation*}
\begin{aligned}
&\hspace{5mm}\displaystyle\lim_{\epsilon\rightarrow +0}\Delta_{T/T_{0}}(x+\im\epsilon)
\\
&=\begin{cases}\displaystyle\frac{1}{2(1-P)\sqrt{x}}\left[-2P\sqrt{x}+\sqrt{\tau_{+}(x)^2-4}-\sqrt{\tau_{-}(x)^2-4}+\im 0\right],\quad 0<x<4(m_{P}-1)^2,
\vspace{3mm}
\\
\displaystyle\frac{1}{2(1-P)\sqrt{x}}\left[-2P\sqrt{x}+\sqrt{\tau_{+}(x)^2-4}+\im\sqrt{4-\tau_{-}(x)^2}\right],\quad 4(1-m_{P})^2<x<4(1+m_{P})^2.
\end{cases}
\end{aligned}
\end{equation*}
\end{enumerate}
Thus the claim follows.
\end{proof}
\begin{Lem}\label{xi:positivezero}
For any $x>4(1+m_{P})^2,$ we have $\lim_{\epsilon\rightarrow+0}\Delta_{T/T_{0}}(x+\im\epsilon)>0.$ In particular, $\xi(x)=0$ for any $x>4(1+m_{P})^2.$
\begin{proof}
By Lemma \ref{lem:taupm}, it is seen that $\tau_{+}(x)>2$ and $\tau_{-}(x)<-2$ for any $P\in(-1, 1).$ Thus, we have
\begin{align}\label{limitzero}
\displaystyle\lim_{\epsilon\rightarrow+0}\Delta_{T/T_{0}}(x+\im\epsilon)=\displaystyle\frac{1}{2(1-P)\sqrt{x}}\left[-2P\sqrt{x}+\sqrt{(2m_{P}+\sqrt{x})^2-4}+\sqrt{(2m_{P}-\sqrt{x})^2-4}\right].
\end{align}
Since the denominator of \eqref{limitzero} is positive, we consider the numerator. We parametrize $x$ as $x=4(1+m_{P}+t)^2$ with $t>0.$ Then we have
\begin{align*}
&\ -2P\sqrt{x}+\sqrt{(2m_{P}+\sqrt{x})^2-4}+\sqrt{(2m_{P}-\sqrt{x})^2-4}
\\
&=-4P(1+m_{P}+t)+2\sqrt{(2m_{P}+t+2)(2m_{p}+t)}+2\sqrt{t(t+2)}
\\
&=2\{\sqrt{(2m_{P}+t+2)(2m_{p}+t)}+\sqrt{t(t+2)}-2P(1+m_{P}+t)\}
\\
&=2\displaystyle\frac{(2m_{P}+t+2)(2m_{P}+t)+t(t+2)+2\sqrt{(2m_{P}+t+2)(2m_{P}+t)}\sqrt{t(t+2)}-4P^2(1+m_{P}+t)^2}{\sqrt{(2m_{P}+t+2)(2m_{p}+t)}+\sqrt{t(t+2)}+2P(1+m_{P}+t)}.
\end{align*}
The numerator of the above fraction is estimated as follows:
\begin{align*}
&\ (2m_{P}+t+2)(2m_{P}+t)+t(t+2)+2\sqrt{(2m_{P}+t+2)(2m_{P}+t)}\sqrt{t(t+2)}-4P^2(1+m_{P}+t)^2
\\
&\ge  (2m_{P}+t+2)(2m_{P}+t)+t(t+2)+2\sqrt{(2m_{P}+t+0)(2m_{P}+t)}\sqrt{t(t+0)}-4P^2(1+m_{P}+t)^2
\\
&= (2m_{P}+t)^2+2(2m_{P}+t)+t(t+2)+2(2m_{P}+t)t-4P^2(1+m_{P}+t)^2
\\
&= 4m_{P}^2+4m_{P}t+t^2+4m_{P}+2t+t^2+2t+2(2m_{P}+t)t-4P^2(1+m_{P})^2-8P^2t(1+m_{P})-4P^2t^2
\\
&=\left\{4m_{P}^2+4m_{P}-4P^2(1+m_{P})^2\right\}+\{8m_{P}t-8P^2(1+m_{P})t\}+4t^2-4P^2t^2+4t
\\
&=4\{m_{P}^2+m_{P}-P^2(1+m_{P})^2\}+8t\{m_{P}-P^2(1+m_{P})\}+4t^2(1-P^2)+4t
\\
&=4t^2(1-P^2)+4t>0,
\end{align*}
where we used 
$$ m_{P}^2+m_{P}-P^2(1+m_{P})^2=0,\quad m_{P}-P^2(1+m_{P})=0.$$
Thus the fraction in \eqref{limitzero} is positive, This implies $\xi(x)=0$ for $x>4(1+m_{P})^2$
\end{proof}
\end{Lem}

\section{Spectral shift function and the Witten index}
\label{section: section 5}
Based on Remark \ref{xi:positive}, Lemma \ref{det:limzero} and Lemma \ref{xi:positivezero}, we now compute the spectral shift function.
\begin{Lem}\label{lem:ssf}
The following assertions hold:
\begin{enumerate}
\item If $P=0,$ then
\begin{align*}
\xi(x)=\displaystyle\frac{1}{2},\quad 0<x<4.
\end{align*}
\item If $0< |P|< \displaystyle\frac{1}{\sqrt{2}},$ then
\begin{align*}
\xi(x)=\begin{cases} \displaystyle\frac{1}{2}-\displaystyle\frac{1}{\pi}\arctan\displaystyle\frac{-2P\sqrt{x}}{\sqrt{4-(2m_{P}+\sqrt{x})^2}+\sqrt{4-(2m_{P}-\sqrt{x})^2}},&  0<x<4(1-m_{P})^2,
\vspace{3mm}
\\
\displaystyle\frac{1}{2}-\displaystyle\frac{1}{\pi}\arctan\displaystyle\frac{-2P\sqrt{x}+\sqrt{(2m_{P}+\sqrt{x})^2-4}}{\sqrt{4-(2m_{P}-\sqrt{x})^2}}, &  4(1-m_{P})^2<x<4(1+m_{P})^2.
\end{cases}
\end{align*}
\item If $|P|= \displaystyle\frac{1}{\sqrt{2}},$ then
\begin{align*}
\xi(x)=\displaystyle\frac{1}{2}-\displaystyle\frac{1}{\pi}\arctan\displaystyle\frac{-2Px^{1/4}+\sqrt{4+\sqrt{x}}}{\sqrt{4-\sqrt{x}}},\quad  0<x<16.
\end{align*}
\item If $\displaystyle\frac{1}{\sqrt{2}}<|P|<1,$ then
\begin{align*}
\hspace{-10mm}\xi(x)=\begin{cases} \displaystyle\frac{1}{2}-\displaystyle\frac{1}{2}\sgn\left[-2P\sqrt{x}+\sqrt{\left(2m_{P}+\sqrt{x}\right)^2-4}-\sqrt{\left(2m_{P}-\sqrt{x}\right)^2-4}\right],\ & 0<x<4(1-m_{P})^2,
\vspace{3mm}
\\
\displaystyle\frac{1}{2}-\displaystyle\frac{1}{\pi}\arctan\displaystyle\frac{-2P\sqrt{x}+\sqrt{(2m_{P}+\sqrt{x})^2-4}}{\sqrt{4-(2m_{P}-\sqrt{x})^2}}, \ & 4(1-m_{P})^2<x<4(1+m_{P})^2.
\end{cases}
\end{align*}
\end{enumerate}
\end{Lem}
\begin{proof}
We observe that $\Delta_{T/T_{0}}(x+\im\epsilon)\in\C_{+},$ 
which implies that $\lim_{\epsilon\rightarrow+0}\Delta_{T/T_{0}}(x+\im\epsilon)\in \overline{\C_{+}}.$ 
To compute the spectral shift function $\xi,$ we make use of the following formula. If $x + \im y \in \overline{\C}_{+} \setminus \{0\},$ then we have
\begin{equation*}
\mathrm{Arg}(x+\im y)=
\begin{cases}
\displaystyle\frac{\pi}{2}-\arctan\displaystyle\frac{x}{y}, &  y>0, \\
0,&  x>0 \mbox{ and } y=0, \\
\pi, & x<0 \mbox{ and } y=0.
\end{cases}
\end{equation*}
This together with Lemma \ref{det:limzero} and the definition of $\xi,$ the lemma follows except the case $P=0.$ 
If $P=0,$  we have
\begin{align}
\displaystyle\lim_{\epsilon\rightarrow+0}\Delta_{T/T_{0}}(x+\im\epsilon)=\displaystyle\frac{\im}{\sqrt{x}}\sqrt{4-x},\quad 0<x<4.
\end{align}
Therefore, $\xi(x)=1/2.$
\end{proof}

\begin{Lem}
We have
\begin{align}\label{index:TT0}
\displaystyle\lim_{t\rightarrow+\infty}e^{-t}\Tr(e^{t\alpha_{P}T(P)}-e^{t\alpha_{P}T_{0}(P)})=\xi(\alpha_{P}^{-1}-0)=
\begin{cases}
0,  &-1<P<0, \\ 
\displaystyle\frac{1}{2},  & 0\le P<1.
\end{cases}
\end{align}
\end{Lem}
\begin{proof}
Let us recall $\alpha_{P}^{-1}=4(1+m_{P})^2$ (see \eqref{alphap-mp}). First, we prove
\begin{align}\label{ssf-edge}
\lim_{t\rightarrow\infty}e^{-t}\Tr(e^{t\alpha_{P}T(P)}-e^{t\alpha_{P}T_{0}(P)})=\xi(\alpha_{P}^{-1}-0).
\end{align}
Let $\epsilon>0$ be sufficiently small so that $(\alpha_{P}^{-1}-\epsilon, \alpha_{P}^{-1})\subset (4(1-m_{P})^2, \alpha_{P}^{-1}).$ Then, by the Krein trace formula (see \cite[Section 9.7]{MR2953553}), it follows that
\begin{align*}
e^{-t}\Tr(e^{t\alpha_{P}T(P)}-e^{t\alpha_{P}T_{0}(P)})&=\displaystyle\int_{\R}e^{-t}(e^{t\alpha_{P}x})'\xi(x)dx
\\
&=\displaystyle\int_{0}^{\alpha_{P}^{-1}}e^{-t}(e^{t\alpha_{P}x})'\xi(x)dx
\\
&=\displaystyle\int_{0}^{\alpha_{P}^{-1}-\epsilon}e^{-t}(e^{t\alpha_{P}x})'\xi(x)dx+\displaystyle\int_{\alpha_{P}^{-1}-\epsilon}^{\alpha_{P}^{-1}}e^{-t}(e^{t\alpha_{P}x})'\xi(x)dx
\\
&=:\mathrm{I}(t, \epsilon)+\mathrm{II}(t, \epsilon).
\end{align*}
Since $\xi$ is piecewise continuous and is bounded on $\R$ $\mathrm{a.e.},$ the dominated convergence theorem yields
$$\displaystyle\lim_{t\rightarrow\infty}\mathrm{I}(t, \epsilon)=0.$$ For $\mathrm{II(t, \epsilon)},$ we have
\begin{align}
\mathrm{II}(t, \epsilon)&= \left[e^{-t}e^{t\alpha_{P}x}\xi(x)\right]_{\alpha_{P}^{-1}-\epsilon}^{\alpha_{P}^{-1}}-\displaystyle\int_{\alpha_{P}^{-1}-\epsilon}^{\alpha_{P}^{-1}}e^{-t}e^{t\alpha_{P}x}\xi'(x)dx
\\
&=\xi(\alpha_{P}^{-1}-0)-e^{-t\epsilon}\xi(\alpha_{P}^{-1}-\epsilon+0)-\displaystyle\int_{\alpha_{P}^{-1}-\epsilon}^{\alpha_{P}^{-1}}e^{-t}e^{t\alpha_{P}x}\xi'(x)dx.
\end{align}
Since $\xi(x)$ is differentiable in $(\alpha_{P}^{-1}-\epsilon, \alpha_{P}^{-1})$ and $\xi'(x)$ is $L^{1}$-integrable in $(\alpha_{P}^{-1}-\epsilon, \alpha_{P}^{-1}),$ the dominated convergence theorem implies
$$ \displaystyle\lim_{t\rightarrow\infty} \mathrm{II}(t, \epsilon)=\xi(\alpha_{P}^{-1}-0).$$
Thus \eqref{ssf-edge} holds.

Next, we prove the right hand side of \eqref{index:TT0}. If $-1<P<0,$ then numerators of composents of $\arctan(\cdot)$ in Lemma \ref{lem:ssf} take positive values. Thus, we have $\xi(\alpha_{P}^{-1}-0)=0.$ If $P=0,$ we have $\xi(4-0)=1/2$ from Lemma \ref{lem:ssf}. If $0<P<1,$ we need some calculation. Similar to Lemma \ref{xi:positivezero}, we parametrize $x$ as $x=4(1+m_{P}-t)^2$ with sufficiently small $t>0.$ Then, we have
\begin{align*}
\displaystyle\frac{-2P\sqrt{x}+\sqrt{(2m_{P}+\sqrt{x})^2-4}}{\sqrt{4-(2m_{P}-\sqrt{x})^2}}&=\displaystyle\frac{-2P(1+m_{P}-t)+\sqrt{(2m_{P}-t)(2m_{P}+2-t)}}{\sqrt{t(2-t)}}
\\
&=\displaystyle\frac{(2m_{P}-t)(2m_{P}+2-t)-4P^2(1+m_{P}-t)^2}{\sqrt{t(2-t)}[2P(1+m_{P}-t)+\sqrt{(2m_{P}-t)(2m_{P}+2-t)}]}.
\end{align*} 
The numerator of above fraction is calculated as follows:
\begin{align*}
& \ (2m_{P}-t)(2m_{P}+2-t)-4P^2(1+m_{P}-t)^2
\\
&=(2m_{P}-t)^2+2(2m_{P}-t)-4P^2(1+m_{P})^2+8P^2t(1+m_{P})-4P^2t^2
\\
&=4m_{P}^2-4m_{P}t+t^2+4m_{P}-2t-4P^2(1+m_{P})^2+8P^2t(1+m_{P})-4P^2t^2
\\
&=\left\{4m_{P}^2+4m_{P}-4P^2(1+m_{P})^2\right\}-4m_{P}t+t^2-2t+8P^2t(1+m_{P})-4P^2t^2
\\
&=0+t\{-4m_{P}+t-2+8P^2(1+m_{P})-4P^2t\}
\\
&=t\{(1-4P^2)t-4m_{P}-2+8P^2(1+m_{P})\}
\\
&=t\left\{(1-4P^2)t-\displaystyle\frac{4P^2}{1-P^2}-2+8P^2\left(1+\displaystyle\frac{P^2}{1-P^2}\right)\right\}
\\
&=t\left\{(1-4P^2)t+2\displaystyle\frac{3P^2-1}{1-P^2}\right\}
\end{align*}
Thus, it follows that
\begin{align*}
&\displaystyle\frac{-2P\sqrt{x}+\sqrt{(2m_{P}+\sqrt{x})^2-4}}{\sqrt{4-(2m_{P}-\sqrt{x})^2}}
\\
&=\displaystyle\frac{\sqrt{t}\left\{(1-4P^2)t+2\displaystyle\frac{3P^2-1}{1-P^2}\right\}}{\sqrt{(2-t)}[2P(1+m_{P}-t)+\sqrt{(2m_{P}-t)(2m_{P}+2-t)}]}\rightarrow 0\quad t\rightarrow+0.
\end{align*}
This implies that
\[ 
\displaystyle\lim_{x\rightarrow 4(1+m_{P})^2-0}
\arctan 
\left[\displaystyle\frac{-2P\sqrt{x}+\sqrt{(2m_{P}+\sqrt{x})^2-4}}{\sqrt{4-(2m_{P}-\sqrt{x})^2}}\right]=0.
\]
By Lemma \ref{lem:ssf}, the results follow.
\end{proof}
\begin{Lem}\label{lem:index}
We have
\begin{align*}
\begin{cases}
w(F_{+}(v^{\ast}))=W(a_+,p_+),
\quad \text{if }|a_{+}|=|p_{+}|<1,
\\
w(F_{-}(v))=-W(a_-,p_-),
\quad 
\text{if }|a_{-}|=|p_{-}|<1.
\end{cases}
\end{align*}
\end{Lem}
\begin{proof}
We only consider $F_{+}(v^{\ast}).$ The other case is similar. This proof is a combination of Table 1, Remark \ref{symmetry}, \eqref{def:T}, \eqref{def:T0} and \eqref{index:TT0}. We recall that
\begin{align*}
F_{+}(v^{\ast})^{\ast}F_{+}(v^{\ast})&=-\alpha_{p_{+}}T(p_{+})+1,\quad F_{+}(v^{\ast})F_{+}(v^{\ast})^{\ast}=-\alpha_{-p_{+}}T(-p_{+})+1.
\end{align*}
Thus, it follows that
\begin{align*}
w(F_{+}(v^{\ast}))=&\displaystyle\lim_{t\rightarrow\infty}\Tr(e^{-tF_{+}(v^{\ast})^{\ast}F_{+}(v^{\ast})}-e^{-tF_{+}(v^{\ast})F_{+}(v^{\ast})^{\ast}})
\\
=&\displaystyle\lim_{t\rightarrow\infty}e^{-t}\Tr(e^{t\alpha_{p_{+}}T(p_{+})}-e^{t\alpha_{-p_{+}}T(-p_{+})})
\\
=&\displaystyle\lim_{t\rightarrow\infty}e^{-t}\Tr(e^{t\alpha_{p_{+}}T(p_{+})}-e^{t\alpha_{p_{+}}T_{0}(p_{+})})-\displaystyle\lim_{t\rightarrow\infty}e^{-t}\Tr(e^{t\alpha_{-p_{+}}T(-p_{+})}-e^{t\alpha_{-p_{+}}T_{0}(-p_{+})})
\\
&=\begin{cases}
0-\displaystyle\frac{1}{2}=-\displaystyle\frac{1}{2},& -1<p_{+}< 0,
\vspace{2mm}
\\
\displaystyle\frac{1}{2}-\displaystyle\frac{1}{2}=0,&  p_{+}=0,
\vspace{2mm}
\\
\displaystyle\frac{1}{2}-0=\displaystyle\frac{1}{2},&  0<p_{+}<1,
\end{cases}\\
&=\frac{\mathrm{sgn}(p_+)}{2}
=W(a_+,p_+).
\end{align*}
where to get the third line, we used $\alpha_{p_{+}}T_{0}(p_{+})=\alpha_{-p_{+}}T_{0}(-p_{+}).$ 
This completes the proof.
\end{proof}
\begin{proof}[Proof of Theorem \ref{thm:main}]
From Theorem \ref{fredholm}, Theorem \ref{exceptional}, \eqref{index} and Lemma \ref{lem:index}, we get the desired result. 
\end{proof}

\appendix
\section{A difference of two exponentials of a bounded operator}
\label{section: Appendix A}
In this section, we prove the following proposition:
\begin{Prop}
We take two operators $H_0, H_{1}\in\mathcal{B}(\mathcal{X}).$ If $H_1-H_0\in \mathcal{S}_{1}(\mathcal{X}),$ then for any $z\in\C,$ we have $e^{zH_{1}}-e^{zH_{0}}\in\mathcal{S}_{1}(\mathcal{X}).$
\end{Prop}
\begin{proof}
For a bounded operator $A,$ we denote by $\|A\|$ the operator norm of $A.$
Similarly, for a trace class operator $A$, we denote by $\|A\|_1$ the trace norm of $A.$ 
We set $M:=\max\{\|H_{1}\|, \|H_{0}\|\}.$ The proof is divided into two steps. First, we show that for any $n\in\N,$ $H_{1}^n-H_{0}^n\in\mathcal{S}_{1}(\mathcal{X})$ and $\|H_{1}^{n}-H_{0}^{n}\|_{1}\le n\|H_{1}-H_{0}\|_{1}M^{n-1}.$ When $n=1,$  this claim is obvious. Suppose that when $n=k,$ the claim follows. Then $n=k+1,$ it follows that
\begin{align*}
H_{1}^{k+1}-H_{0}^{k+1}&=H_{1}H_{1}^{k}-H_{1}H_{0}^k +H_{1}H_{0}^k-H_{0}H_{0}^k
\\
&=H_{1}(H_{1}^k-H_{0}^k)+(H_{1}-H_{0})H_{0}^k\in\mathcal{S}_{1}(\mathcal{X})
\end{align*}
and
\begin{align*}
\|H_{1}^{k+1}-H_{0}^{k+1}\|_{1}&\le\|H_{1}\|\|H_{1}^k-H_{0}^k\|_{1}+\|H_{1}-H_{0}\|_{1}\|H_{0}^k\|
\\
&\le M\times k\|H_{1}-H_{0}\|_{1}M^{k-1} +\|H_{1}-H_{0}\|_{1}M^{k-1}
\\
&=(k+1)\|H_{1}-H_{0}\|_{1}M^{(k+1)-1}.
\end{align*}
Thus the claim follows by induction. Next, we show that $e^{zH_{1}}-e^{zH_{0}}\in\mathcal{S}_{1}(\mathcal{X}).$ First, we observe that
\begin{align}
\displaystyle\sum_{n=1}^{\infty}\frac{|z|^n}{n!}\|H_{1}^n-H_{0}^n\|_{1}\le \displaystyle\sum_{n=1}^{\infty}\frac{|z|^n}{(n-1)!}M^{n-1}\|H_{1}-H_{0}\|_{1}=|z|e^{|z|M}\|H_{1}-H_{0}\|<\infty.
\end{align}
Thus the series $$ \displaystyle\sum_{n=0}^{\infty}\left(\displaystyle\frac{1}{n!}(zH_{1})^n-\displaystyle\frac{1}{n!}(zH_{0})^n\right) $$ converges in $\mathcal{S}_{1}(\mathcal{X}).$ On the other hand, it follow that $$ e^{zH_{j}}=\displaystyle\sum_{n=0}^{\infty} \frac{1}{n!}(zH_{j})^{n}$$
in the operator norm.
Since a convergence in $\mathcal{S}_{1}(\mathcal{X})$ implies a convergence in $\mathcal{B}(\mathcal{X}),$ we have
$$ e^{zH_{1}}-e^{zH_{0}}=\displaystyle\sum_{n=0}^{\infty}\left(\displaystyle\frac{1}{n!}(zH_{1})^{n}-\displaystyle\frac{1}{n!}(zH_{0})^n\right)\in\mathcal{S}_{1}(\mathcal{X}).$$
This completes the proof.
\end{proof}

\section{Spectral decomposition of the discrete Laplacian on \texorpdfstring{$\Z_{\ge0}$}{Z}}
\label{section: Appendix B}
In this section, we consider the unitary transformation from $\ell^2(\Z_{\ge0})$ to $L^2((-2, 2), dt).$ 
This is also discussed in \cite[Example 2 in Section 3.1]{MR2497822}. 
We use the same notations used in the main text. Namely, we set $v$ as a right-shift operator on $\ell^2(\Z_{\ge0})$ and $\{\delta_{x}\}_{x\ge0}$ as the standard basis of $\ell^2(\Z_{\ge0})$ (see \eqref{halfline}). 
\begin{Prop}\label{section B main}
There exists a unitary operator $u: \ell^2(\Z_{\geq0})\rightarrow L^{2}((-2, 2), dt)$ such that $$ u(v+v^{\ast})u^{-1}=M_{t}, \quad (u\delta_{0})(t)=\displaystyle\frac{(4-t^{2})^\frac{1}{4}}{\sqrt{2\pi}}, \quad t\in(-2, 2), $$ where $M_{t}$ is the multiplication operator by $t.$
\end{Prop}

\begin{proof}
First, we note that the set $\{\sqrt{2/\pi}\sin(k(x+1))\}_{x\ge 0}$ is a CONS of $L^{2}([0, \pi], dk).$ We introduce a unitary operator $\F$ by 
$$ (\F\delta_{x})(k):=\sqrt{\displaystyle\frac{2}{\pi}}\sin(k(x+1)),\quad x\in\Z_{\ge0}, \quad k\in[0, \pi]. $$ 
Then, for any $x\in\Z_{\ge0},$ it follows that
\begin{align*}
(\F(v+v^{\ast})\delta_{x})(k)&=\sqrt{\displaystyle\frac{2}{\pi}}\sin(k(x+2))+\sqrt{\displaystyle\frac{2}{\pi}}\sin(kx)
=\sqrt{\displaystyle\frac{2}{\pi}}2 \sin(k(x+1))\cos k\\
&=2\cos k\times\sqrt{\displaystyle\frac{2}{\pi}}\sin(k(x+1))
=2\cos k(\F\delta_{x})(k).
\end{align*}
Thus, we have the operator equality 
$$ \F(v+v^{\ast})\mathcal{F}^{\ast}=2\cos(\cdot).$$
By the definition of $\F,$  we get $(\F\delta_{0})(k)=\sqrt{2/\pi}\sin k.$

Next, we introduce a operator $V: L^{2}([0, \pi], dt)\rightarrow L^{2}((-2, 2), dt)$ by
$$ (Vf)(t):=\displaystyle\frac{1}{(4-t^2)^{\frac{1}{4}}}f\left(\cos^{-1}\displaystyle\frac{t}{2}\right),\quad f\in L^{2}([0, \pi], dt), \quad t\in(-2, 2),$$
where $\cos^{-1}:[-1, 1]\rightarrow [0, \pi].$ By the change of variables $t=2\cos k,$ we have
$$ \displaystyle\int_{-2}^{2}|(Vf)(t)|^{2}dt=\displaystyle\int_{0}^{\pi}\displaystyle\frac{|f(k)|^2}{2\sin k}\times 2\sin k dk=\displaystyle\int_{0}^{\pi}|f(k)|^2dk.$$
Thus, $V$ is isometry. We also introduce a operator $U:L^{2}((-2, 2), dt)\rightarrow L^{2}([0, \pi], dt)$ by
$$ (Ug)(k):=\sqrt{2\sin k}g(2\cos k),\quad g\in L^{2}((-2, 2), dt), \quad k\in [0, \pi].$$
By the change of variables $k=\cos^{-1}\frac{t}{2},$ it is seen that
$$ \displaystyle\int_{0}^{\pi}|(Ug)(k)|^{2}dk=\displaystyle\int_{-2}^{2}2\sin\left(\cos^{-1}\displaystyle\frac{t}{2}\right)|g(t)|^2\times \displaystyle\frac{\frac{1}{2}}{\sqrt{1-\left(\frac{t}{2}\right)^2}} dt =\displaystyle\int_{-2}^{2}|g(t)|^2 dt.$$
Thus $U$ is also isometry. Since $VU=1,$ $V$ is a surjection. Therefore $V$ is unitary. 

For any $f\in L^{2}([0, \pi], dk),$ we have
$$ (V(2\cos k) f)(t)=\displaystyle\frac{1}{(4-t^2)^{\frac{1}{4}}}\times t\times f\left(\cos^{-1}\displaystyle\frac{t}{2}\right)=t(Vf)(t).$$
Hence we get $V(2\cos k)V^{\ast}=M_{t}.$ Moreover, we have
$$ \left(V \sqrt{\displaystyle\frac{2}{\pi}}\sin(\cdot)\right)(t)=\sqrt{\displaystyle\frac{2}{\pi}}\displaystyle\frac{1}{(4-t^2)^{\frac{1}{4}}}\sin\left(\cos^{-1}\displaystyle\frac{t}{2}\right)=\sqrt{\displaystyle\frac{2}{\pi}}\displaystyle\frac{1}{(4-t^2)^{\frac{1}{4}}}\sqrt{1-\left(\displaystyle\frac{t}{2}\right)^2}=\displaystyle\frac{(4-t^{2})^{\frac{1}{4}}}{\sqrt{2\pi}}.$$
Therefore $u:=V\F$ is the desired unitary operator.
\end{proof}

\section{A structure of solutions of a quadratic equation}
\label{section: Appendix C}
In this section, we consider a quadratic equation $w^2+zw+1=0$ parametrized by $z\in\C\setminus[-2, 2].$ A general form of this equation appears in \cite[Section 2.9]{MR2497822}. Recall that the square root of a complex number is defined by
$$ \sqrt{re^{i\theta}}:=r^{\frac{1}{2}}e^{\im\frac{\theta}{2}},\quad r>0,\quad -\pi<\theta<\pi.$$
For $a\in\R$ and $z\in\C\setminus\{-a\},$ we have
\begin{align}\label{fraction}
\displaystyle\frac{z-a}{z+a}=\displaystyle\frac{(z-a)(\overline{z}+a)}{|z+a|^2}=\displaystyle\frac{|z|^2+2a\im \mathrm{Im}\,z-a^2}{|z+a|^2}.
\end{align}
The above calculation implies that
\[ 
\displaystyle\frac{z-a}{z+a}\in\R_{\le 0}\quad \mbox{ if and only if }\quad z\in\R \mbox{ and } |z|\le |a|.
\]
Thus, we have $(z-2)/(z+2)\in\C\setminus(-\infty, 0]$ if $z\in\C\setminus[-2, 2]$ and we can define $\sqrt{(z-2)/(z+2)}.$ Here, it follows that $\mathrm{Re}\sqrt{(z-2)/(z+2)}>0.$ Under the above preliminaries, we define
\begin{align}
H(z):=\displaystyle\frac{\sqrt{\displaystyle\frac{z-2}{z+2}}-1}{\sqrt{\displaystyle\frac{z-2}{z+2}}+1},\quad z\in\C\setminus[-2, 2].
\end{align}

\begin{Prop} 
\label{theorem: properties of H}
The following assertions hold:
\begin{enumerate}
\item The function $z \longmapsto H(z)$ is holomorphic on $\C\setminus[-2, 2].$
\item We have $H(z)^2+zH(z)+1=0$ for any $z\in\C\setminus[-2, 2].$
\item We have $H(\C_{+})\subset  \C_{+}$ and $H(\C_{-})\subset \C_{-},$ where $\C_{\pm} =\{w\in\C|\ \pm\mathrm{Im}\, w>0\}.$
\item We have $|H(z)|<1$ for any $z\in\C\setminus[-2, 2].$
\item For any $z\in\C\setminus[-2, 2],$ we have $$ \displaystyle\frac{1}{2\pi}\int_{-2}^{2}\frac{\sqrt{4-t^2}}{t-z}dt =H(z).$$
\end{enumerate}
\end{Prop}
\begin{proof}
\begin{enumerate}
\item This is obvious from the definition of $H(z).$
\item First, we note that $\sqrt{(z-2)/(z+2)}\neq 1.$ It is seen that
\begin{align*}
H(z)=H(z)\times \displaystyle\frac{\sqrt{\displaystyle\frac{z-2}{z+2}}-1}{\sqrt{\displaystyle\frac{z-2}{z+2}}-1}=\displaystyle\frac{\displaystyle\frac{z-2}{z+2}-2\sqrt{\displaystyle\frac{z-2}{z+2}}+1}{\displaystyle\frac{z-2}{z+2}-1}
=\displaystyle\frac{-z+(z+2)\sqrt{\displaystyle\frac{z-2}{z+2}}}{2}.
\end{align*}
Therefore, we have 
\begin{align}
H(z)^{2}=\displaystyle\frac{z^2+(z^2-4)-2z(z+2)\sqrt{\displaystyle\frac{z-2}{z+2}}}{4}=-1-z\displaystyle\frac{-z+(z+2)\sqrt{\displaystyle\frac{z-2}{z+2}}}{2}=-1-zH(z).
\end{align}
\item We take $z\in\C_{+}.$ Then \eqref{fraction} implies $(z-2)/(z+2)\in\mathbb{C}_{+}.$ By the definition of the square root, we have $\sqrt{(z-2)/(z+2)}\in\mathbb{C}_{+}.$ From the definition of $H(z)$ and \eqref{fraction}, again, we have $H(z)\in\mathbb{C}_{+}.$ The other case is also similar.
\item We take $w\in\C\setminus\{-1\}$ and write $w=x+\im y$ $(x, y\in\R).$ Since 
$$ \left|\displaystyle\frac{w-1}{w+1}\right|^2=\displaystyle\frac{(x-1)^2+y^2}{(x+1)^2+y^2},$$
it follows that
\[ 
\left|\displaystyle\frac{w-1}{w+1}\right|<1\quad \mbox{ if and only if }\quad (x-1)^2+y^2<(x+1)^2+y^2\quad
 \mbox{ if and only if }\quad x>0.
\]
Recall that for any $z\in\C\setminus[-2, 2],$ we have $\mathrm{Re}\sqrt{\displaystyle\frac{z-2}{z+2}}>0.$ The above equivalence implies $|H(z)|<1.$
\item 
Since $H(z)^2+zH(z)+1=0$ holds, it follows that $H(z)\not=0$ for all $z\in\C\setminus[-2, 2].$
From the relation between solutions and coefficients, $1/H(z)$ is also a solution of $w^2+zw+1=0.$ 
By the change of  variables $t=2\cos\theta,$ we have
\begin{align}
\displaystyle\frac{1}{2\pi}\displaystyle\int_{-2}^{2}\frac{\sqrt{4-t^2}}{t-z}dt=\displaystyle\frac{2}{\pi}\int_{0}^{\pi}\displaystyle\frac{\sin^2\theta}{2\cos\theta-z}d\theta=\displaystyle\frac{1}{\pi}\int_{-\pi}^{\pi}\frac{\sin^{2}\theta}{2\cos\theta-z}d\theta,
\end{align}
where to get the last equality we used the fact that the integrand is even. Moreover, by the change of  variables $w=e^{\im\theta},$ we have
\begin{align}
\displaystyle\frac{1}{2\pi}\int_{-2}^{2}\frac{\sqrt{4-t^2}}{t-z}dt&=\displaystyle\frac{1}{\pi}\int_{-\pi}^{\pi}\frac{\left(\frac{e^{\im\theta}-e^{-\im\theta}}{2\im}\right)^2}{2\times\frac{e^{\im\theta}+e^{-\im\theta}}{2}-z}d\theta=\displaystyle\frac{\im}{4\pi}\int_{|w|=1}\frac{(w^2-1)^2}{(w^2-zw+1)w^2}dw.
\end{align}
The solutions of $w^2-zw+1=0$ is $H(-z)$ and $1/H(-z).$ 
By the 4-th assertion, the poles of the integrand inside the unit circle are $w=0$ and $w=H(-z).$
Thus, we have
\begin{align}
\displaystyle\frac{\im}{4\pi}\int_{|w|=1}\frac{(w^2-1)^2}{(w^2-zw+1)w^2}dw&=\displaystyle\frac{\im}{4\pi}2\pi\im\left\{\frac{\left(H(-z)^2-1\right)^2}{\left(H(-z)-\frac{1}{H(-z)}\right)H(-z)^2}+z\right\}
\\
&=-\displaystyle\frac{1}{2}\left\{H(-z)-\displaystyle\frac{1}{H(-z)}+z\right\}.
\end{align}
By the relation between solutions and coefficients, we have $H(-z)+1/H(-z)=z.$ This implies 
\begin{align}\label{h-z}
-\displaystyle\frac{1}{H(-z)}=H(-z)-z.
\end{align}
By using \eqref{h-z}, we have $$ -\displaystyle\frac{1}{2}\left\{H(-z)-\displaystyle\frac{1}{H(-z)}+z\right\}=-H(-z).$$ 
If $-H(-z)=H(z)$ is proven, the proof is completed. By the following observation:
$$ 0=H(-z)^2+(-z)H(-z)+1=(-H(-z))^2+z(-H(-z))+1,$$
$-H(-z)$ is a solution of $w^2+zw+1=0.$ Thus $-H(-z)$ equals to either $H(z)$ or $1/H(z).$ Since $|-H(-z)|<1,$ $|H(z)|<1$ and $|1/H(z)|>1,$ we get $-H(-z)=H(z).$
\end{enumerate}
\end{proof}

\section{Elimination of the phase terms}
\label{section: Appendix D}
\begin{Thm}
\label{lemma: NOW transform}
Let $m \in \Z \setminus \{0\},$ and let
\begin{equation}
\label{equation: generalised ssqw}
S := 
\begin{bmatrix}
\alpha_1 & \beta_1 L^m \\
L^{-m}\beta_1^*  & \alpha_1'(\cdot - m)
\end{bmatrix}, \qquad 
C := 
\begin{bmatrix}
\alpha_2 & \beta_2 \\
\beta_2^* & \alpha_2'
\end{bmatrix},
\end{equation}
where $\alpha_1, \alpha'_1, \alpha_2, \alpha'_2$ are bounded $\R$-valued sequences, and where $\beta_1, \beta_2$ are bounded $\C$-valued sequences. Then there exist $\R$-valued sequences $f,g,$ such that the following two equalities hold true:
\begin{align*}
\begin{bmatrix}
e^{-if} & 0 \\
0      & e^{-ig}
\end{bmatrix} 
S 
\begin{bmatrix}
e^{if} & 0 \\
0      & e^{ig}
\end{bmatrix}
&= 
\begin{bmatrix}
\alpha_1 & |\beta_1| L^m \\
L^{-m}|\beta_1|  & \alpha_1'(\cdot - m)
\end{bmatrix}, \\
\begin{bmatrix}
e^{-if} & 0 \\
0      & e^{-ig}
\end{bmatrix} 
C 
\begin{bmatrix}
e^{if} & 0 \\
0      & e^{ig}
\end{bmatrix}
&= 
\begin{bmatrix}
\alpha_2 & |\beta_2| \\
|\beta_2| & \alpha_2'
\end{bmatrix}.
\end{align*}
\end{Thm}
Note that this unitary transform is strongly based on the method of proof of \cite[Corollary 4.4]{MR4336008},  and we give a direct proof merely for the convenience of the reader.
\begin{proof}
We only prove the case for $m > 0.$ For each $j=1,2,$ let $\theta_j = (\theta_j(x))_{x \in \Z}$ be any $\R$-valued sequence, such that $\beta_j(x) = e^{i \theta_j(x)} |\beta_j(x)|.$  Note that we have
\[
\begin{bmatrix}
e^{-if} & 0 \\
0      & e^{-ig}
\end{bmatrix}
S
\begin{bmatrix}
e^{if} & 0 \\
0      & e^{ig}
\end{bmatrix}
=
\begin{bmatrix}
1 & 0 \\
0 & L^{-y}
\end{bmatrix}
\begin{bmatrix}
e^{-if} & 0 \\
0      & e^{-ig(\cdot + m)}
\end{bmatrix}
\begin{bmatrix}
\alpha_1 & \beta_1 \\
\beta_1^* & \alpha_1'
\end{bmatrix}
\begin{bmatrix}
e^{if} & 0 \\
0      & e^{ig(\cdot + m)}
\end{bmatrix}
\begin{bmatrix}
1 & 0 \\
0 & L^{m}
\end{bmatrix}.
\]
We obtain the following two unitary transforms of the given multiplication operators:
\begin{align}
\begin{bmatrix}
e^{-if} & 0 \\
0      & e^{-ig(\cdot + m)}
\end{bmatrix}
\begin{bmatrix}
\alpha_1 & \beta_1 \\
\beta_1^* & \alpha_1'
\end{bmatrix}
\begin{bmatrix}
e^{if} & 0 \\
0      & e^{ig(\cdot + m)}
\end{bmatrix}
&= 
\begin{bmatrix}
\alpha_1 & |\beta_1| e^{i(\theta_1 + g(\cdot + m) - f)} \\
|\beta_1|e^{-i(\theta_1 + g(\cdot + m) - f)} & \alpha_1'
\end{bmatrix},
\\
\label{equation: transform of coin}
\begin{bmatrix}
e^{-if} & 0 \\
0      & e^{-ig}
\end{bmatrix}
\begin{bmatrix}
\alpha_2 & \beta_2 \\
\beta_2^* & \alpha_2'
\end{bmatrix}
\begin{bmatrix}
e^{if} & 0 \\
0      & e^{ig}
\end{bmatrix}
&=
\begin{bmatrix}
\alpha_2 & |\beta_2| e^{i(\theta_2 + g - f)} \\
|\beta_2|e^{-i(\theta_2 + g - f)} & \alpha_2'
\end{bmatrix}.
\end{align}
The unitary transform \eqref{equation: transform of coin} motivates us to define $g := f - \theta_2.$ It remains to define $f$ in such a way that $\theta_1 + g(\cdot + m) - f = 0$ holds true. If we let $\phi := \theta_2(\cdot + m) - \theta_1,$ then this equality is equivalent to
\begin{equation}
\label{equation: required to show}
f(x + m)  - f(x) = \phi(x), \qquad  x \in \Z.
\end{equation}
For each $x \in \Z,$ we consider $Z_x := \{mx + 0, \dots, mx + (m-1)\}$ consisting of $m$ integers. It is obvious that $\Z$ partitions into the disjoint union $\Z = \bigcup_{x \in \Z} Z_x.$ We let $f(n):= 0$ for each $n \in Z_0.$  Note that any arbitrary number in $\Z \setminus Z_0$ can be uniquely written as $mx + n,$ where $x \in \Z \setminus \{0\},$ and where $n \in Z_0.$ This allows us to let 
\begin{equation}
f(mx + n) := 
\begin{cases}
+\sum_{y=0}^{x-1} \phi(my + n), & x \geq 1, \\
-\sum_{y=1}^{-x} \phi(-my + n), & x \leq -1.
\end{cases}
\end{equation}
Let us first prove that \eqref{equation: required to show} holds true on $Z_{-1} \cup Z_0.$ If $n \in Z_0,$ then $f(n) = 0$ by construction, and so
\begin{align*}
&f(n + m)  - f(n) = f(n + m)  - 0 =  f(m \times 1 + n) = \phi(m \times 0 + n) = \phi(n), \\
&f((n-m) + m)  - f(n-m) = f(n)  - f(n-m) = 0 - f(m \times (-1) + n) = \phi(m \times (-1) + n) = \phi(n - m),
\end{align*}
where $n -m$ belongs to $Z_{-1} = Z_0 - m.$ Let $x' \notin Z_{-1} \cup Z_0.$ On one hand, if $x \geq 1$ and if $x' = mx + n,$ then
\[
f(x' + m)  - f(x') = 
f(m(x+1) + n)  - f(mx + n)
= \sum_{y=0}^{x} \phi(my + n) - \sum_{y=0}^{x-1} \phi(my + n)
= \phi(mx + n) = \phi(x').
\]
On the other hand, if $x \leq -2$ and if $x' = mx + n,$ then
\[
f(x' + m)  - f(x') = 
f(m(x+1) + n)  - f(mx + n)
= -\sum_{y=1}^{-x-1} \phi(-my + n) + \sum_{y=1}^{-x} \phi(-my + n)
= \phi(mx + n) = \phi(x').
\]
This completes the proof.
\end{proof}


\section*{Acknowledgments}
Y.T. acknowledges support by JSPS KAKENHI Grant Number 20J22684. K.~W.~acknowledges support by JSPS KAKENHI Grant Number 21K13846. This work was also partially supported by the Research Institute for Mathematical Sciences, an International Joint Usage/Research Center located in Kyoto University.

\bibliographystyle{plain}
\bibliography{References.bib}

Yasumichi Matsuzawa
\\
Department of Mathematics, 
\\
Faculty of Education, 
\\
Shinshu University, 
\\
6-Ro, Nishi-nagano, Nagano 380-8544, Japan.
\\
{\it E-mail Address}: {\tt myasu@shinshu-u.ac.jp}
\\

Akito Suzuki
\\
Division of Mathematics and Physics, 
\\
Faculty of Engineering, 
\\
Shinshu University, 
\\
Wakasato, Nagano 380-8553, Japan.
\\
{\it E-mail Address}: {\tt akito@shinshu-u.ac.jp}
\\

Yohei Tanaka
\\
Department of Mathematics, 
\\
Faculty of Science,
\\
Shinshu University, 
\\
Matsumoto, 390-8621, Japan
\\
{\it E-mail Address}: {\tt tana35@shinshu-u.ac.jp}
\\

Noriaki Teranishi
\\
Department of Mathematics, 
\\
Hokkaido University, 
\\
Sapporo 060-0810, Japan.
\\
{\it E-mail Address}: {\tt teranishi@math.sci.hokudai.ac.jp}
\\

Kazuyuki Wada
\\
Department of General Science and Education,
\\
National Institute of Technology, Hachinohe College, 
\\
Hachinohe 039-1192, Japan.
\\
{\it E-mail Address}: {\tt wada-g@hachinohe.kosen-ac.jp}

\end{document}